\documentclass[pra,twocolumn,superscriptaddress,longbibliography]{revtex4-1}

\newenvironment{addendum}{%
    \setlength{\parindent}{0in}%
    \small%
    \begin{list}{Acknowledgements}{%
        \setlength{\leftmargin}{0in}%
        \setlength{\listparindent}{0in}%
        \setlength{\labelsep}{0em}%
        \setlength{\labelwidth}{0in}%
        \setlength{\itemsep}{12pt}%
        }
    }
    {\end{list}\normalsize}

\usepackage{amsmath,amsfonts,amssymb,amsthm}
\usepackage{mathtools}
\usepackage{setspace}
\usepackage[normalem]{ulem} 
\usepackage{stmaryrd}       
\usepackage{dsfont}         
\usepackage{expdlist}
\usepackage{physics}
\usepackage{verbatim}
\usepackage{xspace}

\usepackage{graphicx,xcolor}
\usepackage{hyperref}

\hypersetup{pdfstartview=}
\usepackage{url}

\usepackage{verbatim}
\usepackage{alltt}
\usepackage{moreverb}

\usepackage[ruled,lined]{algorithm2e}
\SetAlgorithmName{Protocol}{List of Protocols}

\newcommand{\SI}[1]{\;\mathrm{#1}}

\newcommand{\Prob}{\mathbb{P}}
\newcommand{\Exp}{\mathbb{E}}
\newcommand{\Var}{\mathbb{V}}

\newcommand{\TV}{\mathrm{TV}}

\newcommand{\TMPS}{\mathrm{TMPS}}

\newcommand{\epse}{\epsilon_{\text{en}}}
\newcommand{\epsx}{\epsilon_{\text{ext}}}

\newcommand{\conc}{%
  \mathord{
    \mathchoice
    {\raisebox{1ex}{\scalebox{.7}{$\frown$}}}
    {\raisebox{1ex}{\scalebox{.7}{$\frown$}}}
    {\raisebox{.7ex}{\scalebox{.5}{$\frown$}}}
    {\raisebox{.7ex}{\scalebox{.5}{$\frown$}}}
  }
}

\newcommand{\Sfnt}[1]{\mathbf{#1}} 
\newcommand{\Pfnt}[1]{\mathsf{#1}} 

\newcommand{\cC}{\mathcal{C}}

\newcommand{\cE}{\mathcal{E}}

\newcommand{\cG}{\mathcal{G}}
\newcommand{\cH}{\mathcal{H}}

\newcommand{\cM}{\mathcal{M}}

\newcommand{\cQ}{\mathcal{Q}}

\newcommand{\cT}{\mathcal{T}}

\newcommand{\cX}{\mathcal{X}}

\newtheorem{theorem}{Theorem}
\newtheorem*{theorem*}{Theorem}

\newtheorem*{lemma*}{Lemma}

\newtheorem{proposition}{Proposition}

\allowdisplaybreaks

\begin{document}
\title{Device-independent Randomness Expansion with Entangled Photons}

\author{Lynden K. Shalm}\thanks{L.K.S. (lks@nist.gov) and Y.Z. (yanbaoz@gmail.com) contributed equally to this work.}
\affiliation{Associate of the National Institute of Standards and Technology, Boulder, Colorado 80305, USA}
\affiliation{Department of Physics, University of Colorado, Boulder, Colorado 80309, USA}
\author{Yanbao Zhang}\thanks{L.K.S. (lks@nist.gov) and Y.Z. (yanbaoz@gmail.com) contributed equally to this work.}
\affiliation{NTT Basic Research Laboratories and NTT Research Center for Theoretical Quantum Physics, NTT Corporation, 3-1 Morinosato-Wakamiya, Atsugi, Kanagawa 243-0198, Japan}

\author{Joshua C. Bienfang}
\affiliation{National Institute of Standards and Technology, Gaithersburg, MD 20899, USA}
\author{Collin Schlager}
\author{Martin J. Stevens}
\affiliation{National Institute of Standards and Technology, Boulder, Colorado 80305, USA}
\author{Michael D. Mazurek}
\affiliation{Associate of the National Institute of Standards and Technology, Boulder, Colorado 80305, USA}
\affiliation{JILA, University of Colorado, 440 UCB, Boulder, CO 80309, USA}
\affiliation{Department of Physics, University of Colorado, Boulder, Colorado 80309, USA}

\author{Carlos Abell\'{a}n}
\altaffiliation{Current address: Quside Technologies S.L., C/Esteve Terradas 1, Of. 217, 08860 Castelldefels (Barcelona), Spain}
\affiliation{ICFO-Institut de Ciencies Fotoniques, The Barcelona Institute of Science and Technology, 08860 Castelldefels (Barcelona), Spain}
\author{Waldimar Amaya}
\altaffiliation{Current address: Quside Technologies S.L., C/Esteve Terradas 1, Of. 217, 08860 Castelldefels (Barcelona), Spain}
\affiliation{ICFO-Institut de Ciencies Fotoniques, The Barcelona Institute of Science and Technology, 08860 Castelldefels (Barcelona), Spain}
\author{Morgan W. Mitchell}
\affiliation{ICFO-Institut de Ciencies Fotoniques, The Barcelona Institute of Science and Technology, 08860 Castelldefels (Barcelona), Spain}
\affiliation{ICREA-Instituci\'o Catalana de Recerca i Estudis Avan\c{c}ats, 08010 Barcelona, Spain}

\author{Mohammad A. Alhejji}
\affiliation{JILA, University of Colorado, 440 UCB, Boulder, CO 80309, USA}
\affiliation{Department of Physics, University of Colorado, Boulder, Colorado 80309, USA}
\author{Honghao Fu}
\affiliation{Department of Computer Science, Institute for Advanced Computer Studies, and Joint Center for Quantum \break Information and Computer Science, University of Maryland, College Park, MD 20742, USA}
\author{Joel Ornstein}
\affiliation{Department of Mathematics, University of Colorado, Boulder, Colorado 80309, USA}

\author{Richard P. Mirin}
\affiliation{National Institute of Standards and Technology, Boulder, Colorado 80305, USA}

\author{Sae Woo Nam}
\affiliation{National Institute of Standards and Technology, Boulder, Colorado 80305, USA}

\author{Emanuel Knill}
\affiliation{National Institute of Standards and Technology, Boulder, Colorado 80305, USA}
\affiliation{Center for Theory of Quantum Matter, University of Colorado, Boulder, Colorado 80309, USA}

\maketitle

\SetKwInOut{Given}{Given}
\SetKwInOut{Promise}{Promise}
\SetKwInOut{Access}{Access}
\SetKwInOut{Input}{Input}
\SetKwInOut{Output}{Output}
\SetInd{0.5em}{1em}

\textbf{With the growing availability of experimental loophole-free Bell tests
\cite{hensen:qc2015a,shalm:qc2015a,giustina:qc2015a,rosenfeld:qc2016a, li:2018}, 
it has become possible to implement a new class of \emph{device-independent} 
random number generators whose output can be certified~\cite{colbeck:qc2006a, colbeck:qc2011c} 
to be uniformly random without requiring a detailed model of the quantum devices used 
\cite{bierhorst:qc2018a,liu_yang:qc2018a, zhang_y:qc2018a}.  However, all of these experiments 
require many input bits in order to certify a small number of output bits, and it is  an 
outstanding challenge to develop a system that generates more randomness than is consumed. 
Here, we devise a device-independent spot-checking protocol that consumes  
only uniform bits without requiring any additional bits with a specific bias. 
Implemented with a photonic loophole-free Bell test, we can produce $24\%$ 
more certified output bits ($1,181,264,237$) than consumed input bits ($953,301,640$). 
 The experiment ran for $91.0$ hours, creating randomness at an average rate of $3,606$ bits/s 
 with a soundness error bounded by $5.7\times 10^{-7}$ in the presence of classical side information. Our system will allow for greater trust in public sources of randomness, such as randomness beacons~\cite{nistbeacon}, and  
 may one day enable high-quality private sources of randomness as the device footprint shrinks.}

In 1964, John Bell showed that measurements on entangled quantum systems may show correlations stronger than those  predicted by any local realistic theory~\cite{bell:qc1964a}. In a loophole-free Bell test, entangled particles are 
sent to distant stations, referred to as ``Alice'' and ``Bob'', 
where independent measurements are performed. A violation of local realism occurs if Alice's and Bob's measurements  produce outcomes incompatible with the predictions of any local realistic theory.  In this case,  the measurement outcomes must have some randomness even when conditioned on additional or side information available outside the laboratory. Consequently, in addition to testing whether local realistic theories are consistent with nature, a  loophole-free Bell test can be used to generate uniformly random bits with respect to any adversary isolated from 
the laboratory after the protocol starts~\cite{colbeck:qc2006a, colbeck:qc2011c}. Most importantly, the generated 
random bits can be certified in a device-independent way with a small error.

The first device-independent randomness-generation experiment certified $42$ bits of unextracted randomness in data from a Bell test (subject to the locality loophole) with entangled ions~\cite{pironio:qc2010a} acquired over the course of a month. Since then, major improvements in both experimental design \cite{bierhorst:qc2018a, liu_yang:qc2017a,liu_yang:qc2018a,shen_lijiong:qc2018a, zhang_y:qc2018a} and theoretical analysis \cite{ dupuis:qc2016a,arnon-friedman:qc2018a, knill:qc2018a, zhang_y:qc2020a} have led to remarkable improvements in the achievable bit rate, the minimum time for generating one random bit, the quality of the certificate, and the type of side information an adversary may have. However, these experiments required far more random bits as input during the experiment than generated. For example, in one run of the recent end-to-end protocol we repeatedly
implemented in~\cite{zhang_y:qc2018a}, $k_{\mathrm{out}}=512$ certified output bits were generated from $k_{\mathrm{in}}=4.78\times 10^{7}$ input bits for measurement settings choices and extractor seed---a ratio of $k_{\mathrm{out}}/k_{\mathrm{in}} = 1.07\times 10^{-5}$. A long-standing goal has been to achieve \emph{randomness expansion}, where more certified output bits are generated than input bits consumed. 

\begin{figure*}[htb!]
\begin{center}
\includegraphics[scale=0.42]{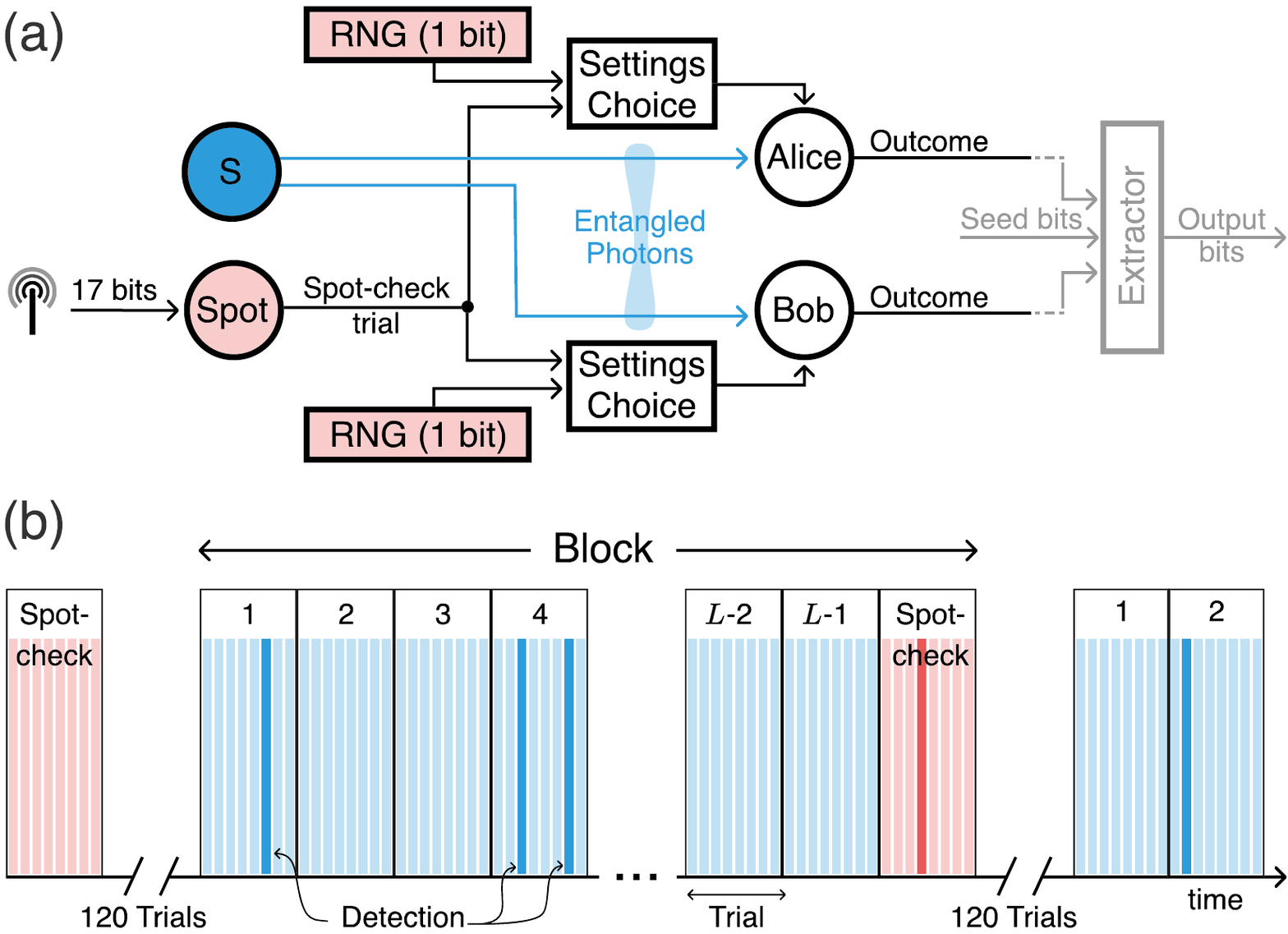}
\end{center}
    \caption{Schematic of the experiment and trial structure. (a) In our protocol, a source station $\Pfnt{S}$ sends entangled photons to Alice and Bob to be measured. At the same station, ``Spot'' randomly signals when it is time to perform a spot-checking trial by counting down from a 17-bit random number $L$.  Alice and Bob have no direct or advance knowledge of when a spot-checking trial will take place, and are unable to communicate with one another. If the next trial is to be a spot-checking trial, then Alice's (Bob's)  settings choice is determined by a random bit from a well-characterized low-latency random number generator (RNG)~\cite{abellan:qc2015a}. Otherwise the settings choice is always fixed to a particular setting.  Alice and Bob measure photons using the settings determined as above, and record their outcomes. If the protocol concludes successfully, then the outcomes along with a small amount of seed randomness can be sent to a classical extractor to extract the output random bits (not implemented). Otherwise, the protocol fails and no new bits are produced. (b) Each block consists of a random number of trials $L$, and each trial is made up of 8 aggregated pulses (light-blue/red bars).  Multiple photons can be detected in any given trial at Alice and Bob (dark-blue/red bars). For the analysis, multiple detection events at Alice (Bob) in a single trial are treated as a single detection event. To allow the Pockels cells, which are part of the measurement setup, to recover after a spot-checking trial, the next 120 trials are ignored before a new block begins. A block always ends with a spot-checking trial.}
  \label{fig:block}
\end{figure*}

\begin{figure*}[htb!]
  \begin{center}
   \includegraphics[scale=0.3]{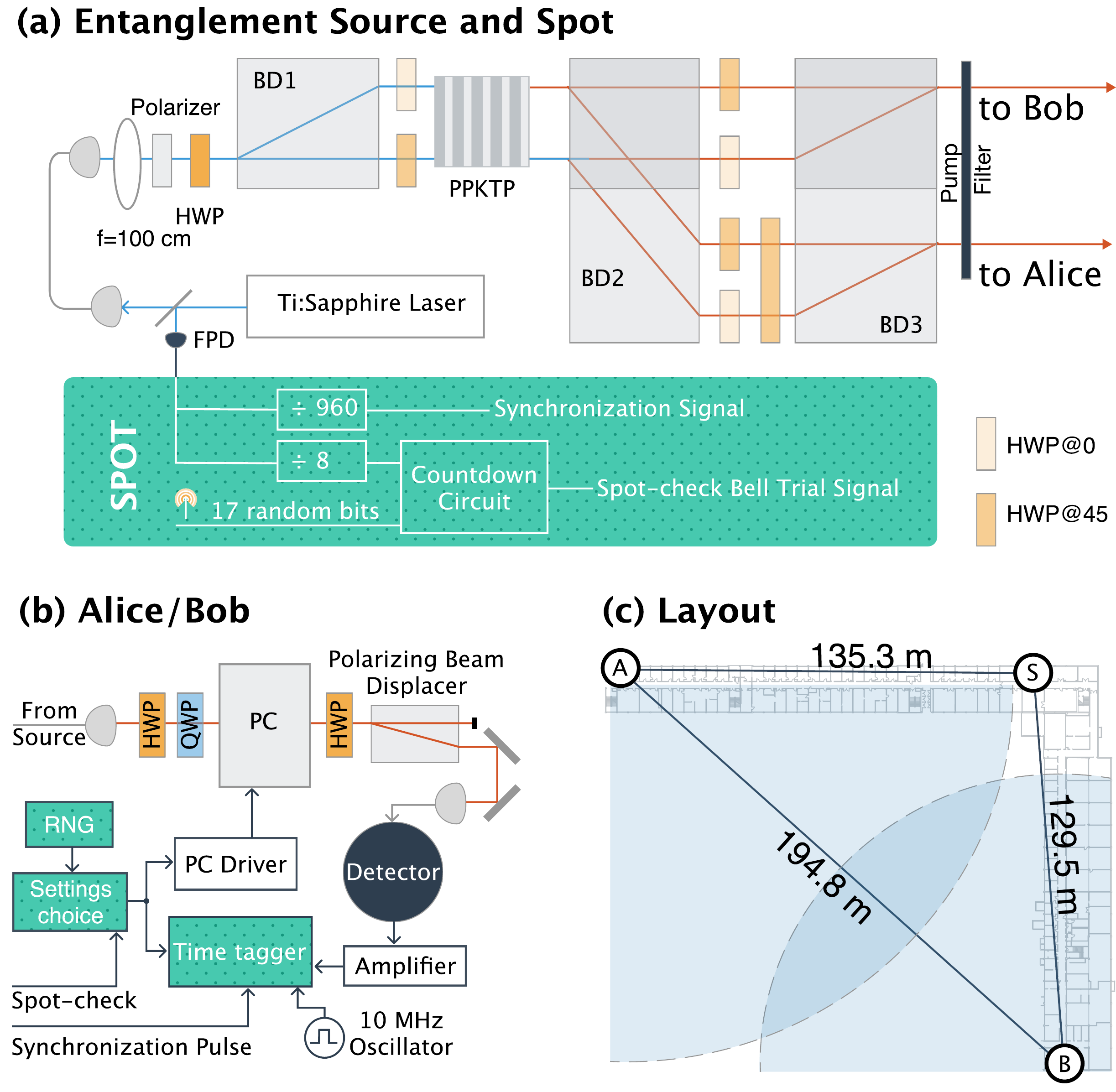}
  \end{center}
      \caption{Source, stations and layout. (a) Our polarization-entangled photons are created by pumping a periodically-poled potassium titanyl phosphate (PPKTP) crystal inside a polarization Mach-Zehnder interferometer made up of beam displacers (BDs) and half-waveplates (HWPs) (see \cite{shalm:qc2015a}) in the state $\ket{\psi} \approx \cos(14.8^{\circ}) \ket{HH} + \sin(14.8^{\circ}) \ket{VV}$. Here $H$ or $V$ represents a horizontally or vertically polarized photon.  The pump is a 775 nm pulsed laser operating at $79.6\SI{MHz}$. A small portion of the pump power is split off and sent to a fast photodiode (FPD) to produce an analog clock signal. This clock signal is used by Spot, who passes it through a divide-by-8 circuit to synchronize with the trials. A countdown circuit takes a $17$-bit number from the NIST randomness 
      beacon~\cite{nistbeacon} to determine when a block ends. At the end of the block, a signal is sent to 
      settings-choice generators at Alice and Bob, informing them to perform a spot-checking trial. A divide-by-960 circuit is also used to provide a synchronizing clock to Alice and Bob to calibrate their time taggers. (b) When the settings-choice generators receive a signal from Spot, they perform a spot-checking trial using a bit from a physical random number generator (RNG) with low latency to control a Pockels cell (PC). With the help of two HWPs and 
      one quarter-wave plate (QWP), Alice chooses between measurement angles $a = -4.1^{\circ}$ and $a' = 25.5^{\circ} $ and Bob chooses between measurement angles $b = 4.1^{\circ}$ and $b' = -25.5^{\circ} $, where the angles are relative to a vertical polarizer. For all non-spot-checking trials the Pockels cells are in settings $a$, $b$. High-efficiency single-photon 
detectors~\cite{Marsili2013} are used to detect the photons, and their arrival times are recorded on a time tagger and saved to a computer. Green boxes shaded with a dot pattern indicate trusted devices. (c) Locations of Alice ($\Pfnt{A}$) and Bob ($\Pfnt{B}$), while the source and Spot are co-located at the station $\Pfnt{S}$. Alice and Bob are located $194.8 \pm 1.0$ m apart. The shaded blue circles represent how far local information about Alice's (Bob's) settings choice could have propagated at the speed of light when Bob (Alice) have completed the measurement for a trial. Alice's and Bob's measurement processes are therefore space-like separated.}
  \label{fig:setup}
\end{figure*}

A key obstacle to achieving randomness expansion in a normal device-independent randomness-generation
setup is the requirement that in every ``trial'' involving a measurement at each station, Alice and Bob must each  uniformly at random choose between two measurement settings, thereby consuming a total of two 
random bits. Instead in our protocol, Alice and Bob only rarely, but randomly, perform a ``spot-checking" 
Bell trial~\cite{pironio:qc2010a, miller_c:qc2014b}. 
 For our spot-checking strategy, we introduce a trusted third party, ``Spot", who decides when Alice and Bob need to perform the spot-checking trial (see figure \ref{fig:block}). Our protocol is based on dividing our experimental trials into blocks of variable length. 
Each block can have a length that ranges anywhere between one trial and $2^{17}$ trials, with the last trial in the block serving as a spot-checking trial as shown in figure \ref{fig:block}.  Spot uses $17$ bits of public randomness
taken from the NIST randomness beacon~\cite{nistbeacon} in order to pick the position of the spot-checking trial 
that ends a block. When the spot-checking trial occurs, Spot sends a signal to Alice and Bob where a settings choice
circuit uses one random bit from a well-characterized low-latency random number generator~\cite{abellan:qc2015a} to
choose each of their measurement settings. Fixed settings are used for the other trials in the block. 
In contrast to the usual spot-checking strategy proposed in the literature (for example, in 
Refs.~\cite{pironio:qc2010a, miller_c:qc2014b}),  our block-wise spot-checking strategy 
 does not require converting uniform bits to non-uniform ones with a specific bias to determine 
when a spot-checking trial is to take place. As neither the adversary nor the untrusted devices used in the 
experiment can learn in advance when a block ends with a spot-checking trial, 
all trials in the block contribute to randomness generation.  It is therefore possible to produce more randomness 
than is consumed. In our experiment, each block is estimated to produce on average $32.80$ bits of randomness 
while consuming a total of $19$ bits, meaning that expansion is possible, see the Supplementary Information 
 for more details. 

In our device-independent setup we do not need to trust the source, and we also do not need to 
trust most of the equipment at Alice and Bob. However, as with all similar cryptographic protocols, the recording 
devices (computers and time taggers) must be trusted. If an adversary had access to these devices before or during 
our experiment, they could compromise the security of our protocol by replacing the experimental records with 
preprogrammed outputs. Moreover,  we assume that Alice's and Bob's  
settings choices used in the spot-checking trial are independent~\cite{bell:qc1993a,shalm:qc2015a}, and that 
the untrusted equipment at Alice, Bob, and the source do not know in advance when a spot-checking trial will occur 
or what the settings choices will be. We ensure that the measurement processes of Alice and Bob are space-like 
separated as shown in figure~\ref{fig:setup} (c).

We use an experimental setup, shown in figure~\ref{fig:setup}, that is based on those reported in 
Refs.~\cite{shalm:qc2015a, bierhorst:qc2018a, zhang_y:qc2018a}.  Due to our low per-pulse probability 
of generating a pair of photons from downconversion, it is advantageous to aggregate $8$ consecutive 
pulses into a single trial. This corresponds to a rate of about $10$ million total trials per 
second and an average of 153 spot-checking trials per second. Spot uses timing information and 17 
random bits to determine which trial corresponds to the end of a block as shown in figure~\ref{fig:setup} (a). 
For all trials except the last spot-checking trial in a block, Alice's and Bob's Pockels cells are turned off. 
This implements settings choice $a$ and $b$ for Alice and Bob respectively as described in 
figure~\ref{fig:setup} (b). Over the course of two weeks we collected $110.3$ hours worth of data 
(see Methods for more details).

Our analysis uses the approach of probability estimation~\cite{zhang:qc2018a,knill:qc2017a}, according 
to which probability estimation factors (PEFs) for each trial are multiplied together such that 
the inverse of the product of the PEFs determines an outcome-probability estimator. Given the set $\cC$ 
of probability distributions of settings choices and outcomes that are achievable by the devices in a given 
trial conditional on classical side information, a PEF with ``power'' $\beta>0$  for this trial is a 
non-negative function $F$ of settings choices and outcomes satisfying a set of linear inequalities imposed 
by each distribution in $\cC$~\cite{zhang:qc2018a,knill:qc2017a}. The coefficients of these inequalities include a $\beta$'th 
power of the settings-choice-conditional outcome probabilities. The power $\beta$ is predetermined and fixed 
for the whole experiment. For an experiment with $n$ total trials, let $F_i$ be the PEF for the $i$'th trial. 
Then $(\prod_{i=1}^{n} F_i\epsilon)^{-1/\beta}$ is an upper bound on the probability of the observed outcomes
conditional on settings choices and classical side information at confidence level $1-\epsilon$.  This can 
be used for randomness certification as explained in Refs.~\cite{zhang:qc2018a,knill:qc2017a}. 
A major advantage of probability estimation is that the resulting 
protocols require significantly less data for randomness generation~\cite{zhang:qc2018a,knill:qc2017a}.  
Because of how it relates to randomness certificates, the quantity $W=\sum_{i=1}^{n}\log_{2}(F_{i})/\beta$ 
is called the running entropy witness up to trial $n$.  The use of PEFs makes it possible to stop the 
experiment early, as soon as the running entropy witness surpasses the corresponding success 
threshold~\cite{zhang:qc2018a,knill:qc2017a}. This early-stop feature was exploited in our analysis, see below.

A randomness-generation protocol takes as input the desired number of random bits $k_{\text{out}}$ and 
a soundness error $\epsilon$.  When it is run, it either fails or succeeds. If it succeeds, it produces
$k_{\text{out}}$ bits. The behavior is such that within a statistical distance of $\epsilon$, there exists 
an ``ideal'' randomness-generation protocol that has the same probability of success and produces $k_{\text{out}}$
uniformly random bits conditional on success. The random bits produced by the ideal protocol are independent of
settings choices and seed input as well as side information. A consequence is that if the probability of success 
is at least $\sqrt{\epsilon}$, then conditioned on success,  the random bits produced by the actual protocol 
are indistinguishable (within a statistical distance of $\sqrt{\epsilon}$) from those produced by the ideal protocol.

Before running our analysis, we set a stopping criteria as discussed in the Methods section. This  
consits of picking a minimum value, $W_{\min}$, that must be met by the running entropy witness
 in order for the protocol to succeed. Figure \ref{fig:primary_yield} shows 
some of the tradeoffs involved in choosing $W_{\min}$. Our choice of $W_{\min} = 1,616,998,677$ 
corresponds to an estimated probability of success of $p_{\mathrm{succ}} \geq 0.9938$. The actual run of the protocol
analysis succeeded after $49,977,714$ blocks ($91.0$ hours of experiment time). At this point, we have the 
ability to generate $1,181,264,237$ new random bits at the soundness error $5.7 \times 10^{-7}$, while 
consuming  $949,576,566$ random bits for spot checks and settings choices.  In terms of 
randomness rate, $3,606$ new random bits can be generated per second.  If we were to extract our bits using 
Trevisan's extractor implemented by Ref.~\cite{mauerer:qc2012a} with parameters as described in 
Ref.~\cite{bierhorst:qc2018a}, an additional $3,725,074$ seed random bits would be required. We are therefore able 
to achieve an expansion ratio of $k_{\text{out}}/k_{\text{in}} =1.24$. In figure \ref{fig:main_anlys} we show the 
running entropy witness on all the data available for expansion analysis.

\begin{figure}[htb!]
  \begin{center}
   \includegraphics[scale=0.62]{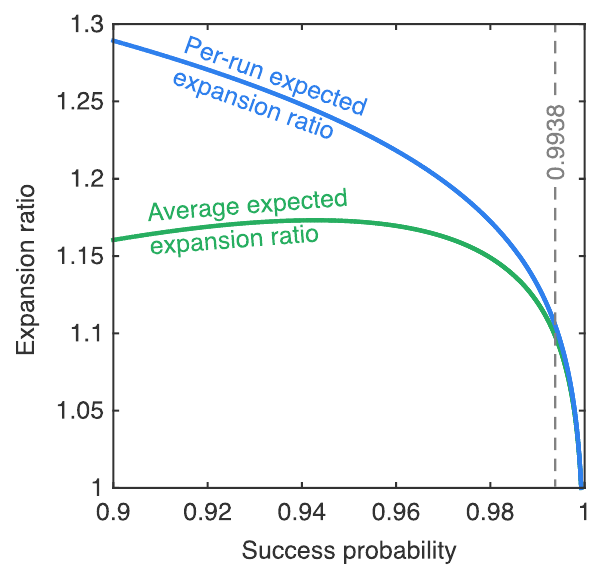}
  \end{center}
      \caption{Tradeoff between the expected expansion ratio and the desired success probability.
       This relation is heuristically determined; for this we 
use the commissioning data and fix the soundness error to be $\epsilon=5.7\times 10^{-7}$ (see the Methods section).   For high expansion ratios, the probability of our protocol succeeding drops. We choose our stopping criteria such that our expected success probability is at least $0.9938$ (the dashed line). Consequently, on a single successful run (the blue curve) we would expect to achieve an expansion ratio of $1.105$. When we ran our protocol, we reached the stopping criteria sooner than expected leading to an actual expansion ratio of 1.24. The green curve estimates how well our device would perform on average if it were run many times (as opposed to a single shot as we report here), where we would expect the protocol to occasionally fail. Protocol failures are costly as they consume random bits but do not produce any randomness. Hence, we expect that our device is still capable of expansion if it is continuously operated with the occasional expected protocol failure.}  
\label{fig:primary_yield}
\end{figure}

\begin{figure}[htb!]
  \begin{center}
   \includegraphics[scale=0.62]{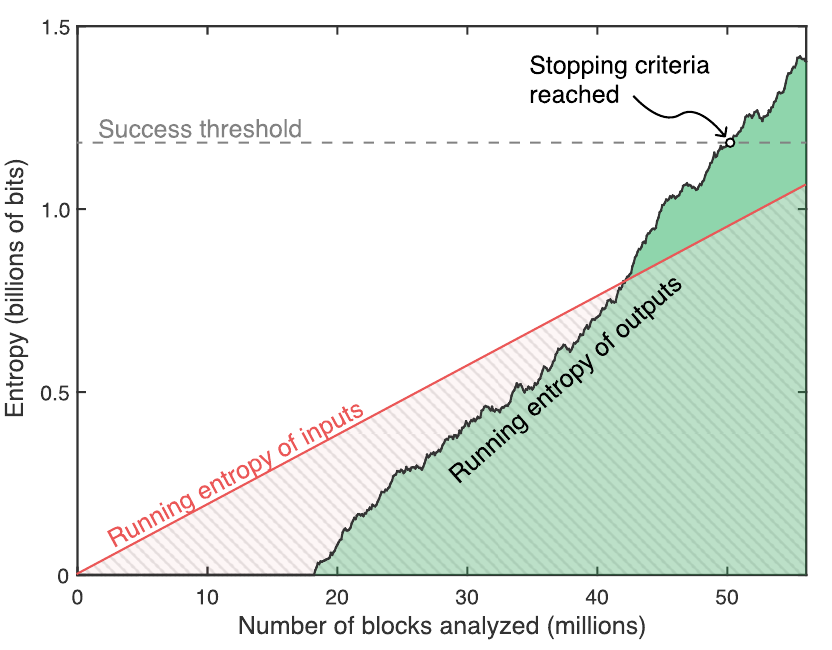}
  \end{center}
      \caption{Entropies as a function of the number of blocks processed in our protocol run. The running entropy of outputs (the black curve) is the running entropy witness adjusted to account for our soundness error of $5.7\times 10^{-7}$. After processing $49,977,714$ blocks (corresponding to $91.0$ hours of experiment time), our success 
threshold shown as the dashed line (corresponding to the $\sigma_{\text{in}}$ in the Methods section) is reached.  
The red curve represents the running total of the consumed input bits  
(including the seed bits needed for the extractor~\cite{mauerer:qc2012a,bierhorst:qc2018a}). We compute the running   entropies for the remaining blocks after reaching our stopping point to study the performance of our system. Because the running entropy witness must be adjusted to account for our soundness error, we must process $18,196,425$ blocks before the running entropy of outputs becomes positive.  When we reach our stopping criteria, we are able to achieve an expansion ratio of $1.24$, having generated at least $1,181,264,237$ new bits after accounting for the required $949,576,566$ random bits for spot checks and settings choices as well as $3,725,074$ random bits needed as a seed. 
   }
     \label{fig:main_anlys}
\end{figure}

Our current analysis is secure against classical side information 
but not quantum side information. Security against quantum side information can be 
obtained by the extension of PEFs to quantum estimation factors (QEFs) as done in 
Refs.~\cite{zhang_y:qc2018a, knill:qc2018a, zhang_y:qc2020a}. We did not do so here 
primarily because our current method for constructing QEFs requires too much 
 computation time for our experiment.  With lower lossess, 
$\beta$ increases significantly, enabling the use of QEFs. Our symmetric system 
efficiency is $76.3 \pm 0.5\%$, and is determined mainly by the efficiency  
of our detectors which is on average $91.0 \pm 1.0\%$.  Increasing the efficiency of our 
detectors to $98.0\%$ would lead to not only an enhancement of security 
against quantum side information, but also an order of magnitude reduction in 
the running time of our experiment. An additional benefit 
of shorter running times is that running an extractor, like Trevisan's 
extractor~\cite{mauerer:qc2012a,bierhorst:qc2018a}, becomes computationally tractable. 
 Trevisan's extractor works in the presence of classical as well as quantum 
side information~\cite{mauerer:qc2012a, de:2009}, with the caveat that its working
parameters, such as the number of seed bits required, may depend on the type of 
side information considered. 
While other accessible extractors are computationally more practical, they require 
more seed bits than the input bits for spot checks and settings choices. 
If the seed bits can be reused as is possible for strong extractors such as Trevisan's 
extractor, then they do not necessarily need to be counted as a consumed resource; 
however, previously used seed bits would have a statistical distance from uniform 
that needs to be added to the soundness error.  

Our current protocol requires a trusted third party, Spot, to determine when 
Alice and Bob need to perform a spot-checking trial. In principle, our protocol can 
be modified such that a trusted third party is not needed. However, the modified
block-wise spot-checking protocol appears to require a much longer running time 
for randomness expansion, an issue which deserves further investigation. 

Due to the size and complexity of a loophole-free Bell test, the first practical 
application of device-independent random number generators will likely be as a source 
of public randomness in randomness beacons. The current NIST randomness beacon operates 
at a rate of 512 bits/min~\cite{nistbeacon}. In our previous work we were able to 
generate 512 bits of certified randomness from approximately 5 minutes of data on 
average~\cite{zhang_y:qc2018a}. In this work, we are able to certify an average of 
$3,606$ bits/s  over the duration of the experiment. It should be noted that 
we are only able to certify these bits after we reach our stopping criteria, so there 
is a large latency involved. However, this does show that device-independent randomness  
generation is now within reach of real-time integration in randomness beacons. 
We used $3.17$ years of bits produced by the NIST randomness beacon, and in  $91.0$ 
hours certified enough randomness to in principle power the beacon for the next $3.93$ 
years, provided that they can be kept secret until broadcasting. Furthermore, with two 
separate device-independent randomness expansion devices it would be possible to greatly 
increase the overall expansion ratio. The expanded output bits from one device could be 
used as input for the other device and vice versa. Starting with a relatively small 
number of uniform input bits, these two systems could ``cross-feed'' one another 
to produce vastly more certified output bits~\cite{Coudron:2014:IRE:2591796.2591873,
miller_c:qc2014a, chung:qc2014a} for use in public sources of randomness. 

Our work can be thought of as implementing a simple, but non-trivial, quantum network. 
Entanglement is distributed and used to perform a task (device-independent randomness expansion) 
that no classical system is capable of performing. This unique quantum advantage arises 
from the nonlocal correlations possible in distributed entangled systems, and demonstrates 
an aspect of the potential power of larger-scale future quantum networks. 

\noindent{\textbf{Methods}}\\ 
\emph{Data Acquisition}: To keep the experiment aligned and well-functioning we collected data in a series of cycles, with each cycle consisting of up to one hour worth of data for expansion analysis. Every cycle began with approximately
 2 minutes of calibration data consisting of a standard loophole-free Bell test with the Pockels cells operating at $250$ kHz, which was stored in a calibration file. Subsequent data was recorded in expansion files, with each expansion file  consisting of $2^{14}$ blocks (approximately $2$ minutes).  At the end of each expansion file, approximately $5$ seconds of additional calibration data was saved.  After every 5 expansion files, a quick check was made using motorized waveplates to see whether the experiment was still performing well. If the efficiency dropped or the visibility of the entangled state changed, then an automated realignment of the setup was performed, and a new cycle was started. If not, up to 30 expansion files were obtained in the current cycle before proceeding to the next cycle. In this way, enough calibration data was collected to allow our analysis protocol to adapt to experimental drifts in our setup. 

\noindent{\emph{Parameter Determination}}: There are several parameters that must be determined before
running our analysis so that the desired number of random bits
$k_{\text{out}}$ at soundness error $\epsilon$ can be obtained.
They are the power $\beta > 0$, a maximum number of blocks $N_b$ to acquire, 
and a minimum final entropy witness $W_{\min}$ required
for success. The analysis stops with success if the running entropy
witness exceeds $W_{\min}$ and fails if this is not achieved after $N_b$
blocks.  To determine these parameters, we use 
the first $16$ cycles of data, which contains $4,502,276$
blocks (about $7.4\%$ of the recorded data), for commissioning and
training purposes. After the parameters are determined and fixed, the protocol is run 
on the non-commissioning data. Blocks are analyzed sequentially 
using the PEFs with power $\beta$. The PEFs are constructed and updated  
using the calibration data that is periodically taken during the experiment, allowing 
our analysis to adapt to any experimental drifts or changes that occur over time 
(see the Supplementary Information for details). 
After processing the blocks in each file, a check is made to see if the running entropy witness 
surpasses $W_{\min}$.

The randomness extraction part of the protocol is based on Trevisan's extractor 
as described in Refs.~\cite{mauerer:qc2012a,bierhorst:qc2018a} and detailed in 
 the Supplementary Information.  
Its parameters are also determined during commissioning. In addition to the desired
number of output bits $k_{\text{out}}$, they include the length (in bits)
of the experimental output string $m_{\text{in}}$, the extractor error 
$\epsilon_{\text{ext}}$ and the number of seed bits $d_{\text{s}}$ required.  
We also need to specify the min-entropy (in bits) of the experimental 
output string required for success $\sigma_{\text{in}}$, which is determined 
by the minimum entropy witness $W_{\min}$, the soundness error $\epsilon$ and 
the extractor error $\epsilon_{\text{ext}}$. 
The actual experimental output string has a variable length, depending on 
the actual lengths of the blocks and when the running entropy witness exceeds
$W_{\min}$. This string is zero-filled to length $m_{\text{in}}$, so 
$m_{\text{in}}$ is an upper bound on the maximum length of the output string. 
Because of the large length of the experimental output string, an implementation 
of the classical Trevisan's extractor would take prohibitively long (many months 
on a large supercomputer) to explicitly extract the desired $k_{\text{out}}$ random 
bits. As a result we did not do so for this demonstration.

For demonstration purposes, and because the protocol analysis was
performed months after the experimental run and data acquisition were
completed, we determined the protocol parameters
from both the commissioning data and knowledge of how many blocks were
acquired during the experiment. After considering the tradeoffs, we 
set our soundness error to be $\epsilon=5.7\times 10^{-7}$, corresponding 
to the $5$-sigma criterion, as this enables good expansion with reasonable security. 
We then constrain the parameters so that a heuristically determined probability 
of success satisfies $p_{\mathrm{succ}} \geq 0.9938$, exceeding the 
conventional one-sided $2.5$-sigma criterion as a compromise between good
completeness and expansion. Based on these criteria, we choose $k_{\text{out}}=1,181,264,237$, 
$W_{\min}=1,616,998,677$, $\beta=4.7614\times 10^{-8}$, and $N_b=56,070,910$, 
where $N_b$ is the actual number of non-commissioning blocks acquired in the experiment.  
The associated extractor parameters are $m_{\text{in}}=14,698,652,631,040$, 
$\epsilon_{\text{ext}}=1.78 \times 10^{-9}$, $d_{\text{s}}=3,725,074$,  
and $\sigma_{\text{in}}=1,181,264,480$.  Based on these choices, 
if all non-commissioning blocks are used for achieving a successful 
randomness expansion, the expected expansion ratio $k_{\mathrm{out}}/k_{\mathrm{in}}$ 
is $1.105$. If success is achieved before processing all blocks, the actual 
expansion ratio improves as fewer input bits are consumed to certify the 
same number of output bits. 
Any random bits used for commissioning or calibration purposes can be 
obtained from pseudorandom sources, and therefore do not count as 
consumed randomness in our analysis. We emphasize that the success 
probability is heuristically determined based on the commissioning data
and under the assumption that the entropy witness accumulated over the 
subsequent non-commissioning blocks for expansion analysis is normally 
distributed. For more details behind parameter determination, see the 
Supplementary Information.

\begin{addendum}
 \item We thank Carl Miller and Scott Glancy for help with reviewing this paper.
  This work includes contributions of the National Institute of
  Standards and Technology, which are not subject to U.S. copyright.
  The use of trade names does not imply endorsement by the
  U.S. government. The work is supported by the National Science Foundation RAISE-TAQS (Award 1839223); European Research Council (ERC) projects AQUMET (280169) and ERIDIAN (713682); European Union projects FET Innovation Launchpad UVALITH (800901);~the Spanish MINECO projects OCARINA (Grant Ref. PGC2018-097056-B-I00) and Q-CLOCKS (PCI2018-092973), the Severo Ochoa programme (SEV-2015-0522); Ag\`{e}ncia de Gesti\'{o} d'Ajuts Universitaris i de Recerca (AGAUR) project (2017-SGR-1354); Fundaci\'{o} Privada Cellex and Generalitat de Catalunya (CERCA program); Quantum Technology Flagship project MACQSIMAL (820393); Marie Sk{\l}odowska-Curie ITN ZULF-NMR (766402); EMPIR project USOQS (17FUN03).
\end{addendum}


\onecolumngrid
\newpage

\renewcommand{\theequation}{S\arabic{equation}}

\newcommand{\boxname}{Box}
\newcounter{box}

\section*{Supplementary Information: Device-independent Randomness Expansion with Entangled Photons}
\label{sec:SM}

This Supplementary Information is structured as follows.  In Sect.~\ref{sect:pe_thry}, we 
introduce the approach of probability estimation for certifying randomness with respect 
to classical side information~\cite{zhang:qc2018a, knill:qc2017a}. Particularly, we 
show how to perform probability estimation and certify randomness  
by means of probability estimation factors (PEFs) with a sequence 
of blocks. As a consequence, we can design an end-to-end protocol for randomness
generation, which is presented in Sect.~\ref{sect:rand_gen}. Details on the numerical 
construction of PEFs are provided in Sect.~\ref{sect:pef_constrct}. Then, we present
our protocol design and data analysis. Specifically, in Sect.~\ref{sect:commission} we 
explain how to determine the protocol parameters before running our randomness-generation 
protocol.  In Sect.~\ref{sect:calib} and Sect.~\ref{sect:anlys}, we explain how we update 
PEFs on the basis of calibration data acquired during the experiment and describe the 
results of running the protocol.

We consider an experiment with a sequence of time-ordered blocks,
where each block consists of at most $2^k$ trials 
executed in time order. Here $k$ is a positive integer. The number of 
trials in a block (that is, the length of a block) is determined by the value $l$ 
of a uniform random variable $L$, where $l\in\{1,...,2^k\}$.   
Thus the average block length is $(1+2^k)/2$.
In our experiment, each trial in a block uses a pair of quantum devices
held by two remote stations, Alice and Bob. 
Before a trial, the state of the quantum devices can be correlated 
with the classical side information $E$ possessed by an adversary.  In each
trial, the binary settings choices $X$ and $Y$ are provided to the
devices of Alice and Bob as the trial inputs, and the corresponding
binary outcomes $A$ and $B$ are obtained from the devices as the trial
outputs. The inputs and outputs of a trial together are called the results 
of the trial.  As is conventional,  values of a random variable 
are denoted by the corresponding lower-case symbol. Thus $x$ is a value
of the random variable $X$, and we have $a,b,x,y\in\{0,1\}$. 
The last trial in a block is called the spot-checking trial. 
For the spot-checking trial, the inputs $X$ and $Y$ are uniformly distributed, while
for the other trials in a block, the inputs are fixed to be $X=0$ and $Y=0$.
In this work,  we denote probability distributions by lower-case Greek letters 
(such as $\mu$ and $\nu$). The expectation and variance of a random variable 
$X$ according to a probability distribution $\mu$ are denoted by $\Exp_{\mu}(X)$ 
and $\Var_{\mu}(X)$, respectively. The probability of an event $\Phi$ according 
to $\mu$ is denoted by $\Prob_{\mu}(\Phi)$.

\section{Probability estimation}
\label{sect:pe_thry}

In this section, we first introduce the main concepts of probability estimation~\cite{zhang:qc2018a, 
knill:qc2017a}, namely models and probability estimation factors (PEFs). Then we explain how PEFs
can certify the smooth conditional min-entropy in the sequence of outputs conditional on the
sequence of inputs as well as the classical side information $E$. 

\subsection{Models and PEFs}
\label{subsect:pefs}
Consider a generic block obtained in the experiment. Let $L$ be
the block-length random variable with value $l\in\{1,...,2^k\}$.  For
the $j$'th trial in the block where $j\in \{1,...,l\}$, the inputs for
Alice and Bob are denoted by $X_{j}$ and $Y_{j}$, and the
corresponding outputs are denoted by $A_{j}$ and $B_{j}$. We use 
$\Sfnt{X}$ and $\Sfnt{Y}$ to denote the sequences of inputs of Alice 
and Bob in the block. Similarly,
$\Sfnt{A}$ and $\Sfnt{B}$ denote the sequences of the corresponding
outputs. For convenience, below we abbreviate the $X_{j}$ and $Y_{j}$
together as $Z_j$ and the $A_{j}$ and $B_{j}$ together as $C_j$. That
is, $Z_j=X_{j}Y_{j}$ and $C_j=A_{j}B_{j}$. We write $\Sfnt{Z}$ and
$\Sfnt{C}$ for the sequence of $Z_{j}$'s and $C_{j}$'s in the block, respectively. 
The $C_j$ and $Z_j$ together are called the results of the $j$'th trial in the block,
and the $\Sfnt{C}$ and $\Sfnt{Z}$ together are called the results of the
block.  As defined so far, the value $l$ of the random variable $L$ is 
the length of these sequences. Therefore, $C_{j}$ and 
$Z_{j}$ are not yet defined for $j>l$, but for randomness generation 
we need to pad the sequences with a string of zeros to 
the maximum length $2^{k}$ (see Sect.~\ref{subsec:protocol}). 
For the analysis below, instead we fill the sequences to the maximum length 
with the special symbol '$*$' so that  $C_{j}=Z_{j}=*$ for $j>l$. 
We call the trials that are actually executed \emph{real} trials and 
the trials whose results are $*$-filled \emph{virtual} trials. So, a block 
has a random number $L$ of real trials as well as $(2^k-L)$ virtual 
trials. 

To certify randomness with respect to the classical side information
$E$, we need to characterize the set of all possible probability
distributions of $\Sfnt{C}\Sfnt{Z}$ achievable by a
block conditionally on each value $e$ of $E$.  That is,
 we need to characterize the set of
all possible probability distributions of $\Sfnt{C}\Sfnt{Z}$ which
satisfy verifiable physical constraints on device behavior and are
achievable by a block. It suffices to characterize a superset $\cC$ 
of this set. This superset is called the model for a block. 
  To describe the model for a block, we
  first consider the case of a single trial in the block.
  The model for each trial $j$ in the block may depend on the past, 
particularly on whether the block has not yet ended as expressed by 
the event $L\geq j$. This trial model is a superset of the set 
of all possible probability distributions of $C_jZ_j$ which satisfy 
verifiable physical constraints and are achievable by the $j$'th 
trial conditionally on the past and classical side information $E$. 
  In our experiment the quantum devices used in each trial are
  constrained only by quantum mechanics and locality. Given this
  and the condition $L\geq j$, we can take the model 
  $\cM_{j}^{(\text{r})}$ for the $j$'th real trial to be the set of probability
  distributions $\mu(C_jZ_j)$ satisfying the following two constraints:
  1) The conditional distributions $\mu(C_j|Z_j)$ with $C_j=A_{j}B_{j}$
  and $Z_j=X_{j}Y_{j}$ satisfy the non-signaling conditions~\cite{PRBox} 
  and Tsirelson's bounds~\cite{Tsirelson:1980}. The set of such conditional 
  distributions is denoted as $\cT$, which is a convex polytope with 
  $80$ extreme points and includes all conditional distributions 
  achievable by quantum mechanics (see Sect.~VIII.~A of Ref.~\cite{knill:qc2017a}). 
  2) The inputs $Z_j=X_{j}Y_{j}$ are chosen according to a fixed distribution
that depends on the trial position in the block as follows: 
  With probability $q_j=1/(2^k-j+1)$ the $j$'th trial is a spot-checking trial and so  
  the distribution of $Z_j$ is uniform, and with probability $(1-q_j)$ the trial 
  is a non-spot-checking trial and so the inputs $Z_j=X_{j}Y_{j}$ are fixed to be $X_j=0$ 
  and $Y_j=0$.  The experiment is configured so that both the quantum devices and side 
  information $E$ do not know which of the two possibilities will occur at the next trial 
  of the block. The only information available to the devices and $E$ is that the
  block has not yet ended. Thus, from the point of view of the
  devices or $E$, conditional on $L\geq j$ the input distribution 
  $\mu(X_jY_j)$  for each member $\mu$ of $\cM_{j}^{(\text{r})}$ satisfies that
$\mu(X_j=0,Y_j=0)=1-3q_j/4$ and $\mu(X_j=x,Y_j=y)=q_j/4$ if $x\neq 0 $ or
$y\neq 0$. So we define the fixed input distribution  $\nu_{j}(X_jY_j)$ according to
\begin{equation} \label{eq:fixed_input_dist}
    \nu_j(X_j=x, Y_j=y)=   
    \left\{\begin{array}{ll}
          1-3q_j/4& \textrm{if $x=y=0$,}\\
          q_j/4 &\textrm{if $x\neq 0$ or $y\neq 0$.}
        \end{array}\right.
  \end{equation}
As the input distribution $\nu_{j}(XY)$ is fixed, the same as $\cT$ the model 
$\cM_{j}^{(\text{r})}$ is a convex polytope with 80 extreme points. 
 We assume that after the last, spot-checking trial of the block, both 
 the quantum devices and $E$ learn that the block has ended. 
Given our definitions, the model $\cM_j^{(\text{v})}$ for the virtual trial 
with $C_{j}=Z_{j}=*$ and $j>l$, where $l$ is the actual block length, becomes 
trivial in the sense that the model has only a fixed and deterministic distribution. 
Depending on whether the condition $L\geq j$ is satisfied or not, the model for the $j$'th trial 
in a block is characterized as $\cM_j^{(\text{r})}$ or $\cM_j^{(\text{v})}$ introduced above. We 
denote the model for the $j$'th trial by $\cM_j$ when it is not specified whether the condition 
$L\geq j$ is satisfied or not. 

To perform probability estimation in order
to certify randomness in $C_j$ conditional on $Z_j$ and $E$, we
introduce probability estimation factors (PEFs) for the trial model
$\cM_j$. A PEF with a positive power $\beta$ is a non-negative
function $F_j: cz\mapsto F_j(cz)$ satisfying the PEF inequality
\begin{equation}\label{eq:pef_def}
  \sum_{cz} \mu(C_j=c, Z_j=z) F_j(cz) \mu(C_j=c|Z_j=z)^{\beta}\leq 1
\end{equation}
for each probability distribution $\mu(C_jZ_j)$ in the trial model $\cM_j$. We note that
 to satisfy the PEF inequality for all distributions in $\cM_j$, it suffices 
 to satisfy this inequality for the extremal distributions of $\cM_j$ according 
 to Lemma~14 of Ref.~\cite{zhang:qc2018a}. Furthermore, the constant function $F(cz)=1$ 
 for all $cz$ is a PEF with power $\beta$ for all models and all $\beta>0$. 
 We choose this constant function as the PEF for each virtual trial.

 Next we can construct the model $\cC$ for a block as a chain of
 models $\cM_j$ for each trial $j$ in the block.  Let $\Sfnt{Z}_{<j}$ and
 $\Sfnt{C}_{<j}$ be the sequences of trial inputs and outputs before
 the $j$'th trial in the block.  Define the random variable $S_{j}$ as 
 $S_{j}=1$ if $L\geq j+1$ and $S_{j}=0$ otherwise. Let $\Sfnt{S}$ 
 denote the sequence of $S_j$'s for the block. Then the sequence $\Sfnt{S}$
 consists of $(L-1)$ consecutive $1$'s followed by $(2^k-L+1)$ consecutive $0$'s.
 Equivalently, the block length $L$ is determined by $\Sfnt{S}$. 
 To specify the distribution of results $\Sfnt{CZ}$ in a block with 
 length $L$, it is equivalent to know the joint distribution $\mu$ of the 
 variables $\Sfnt{C}$, $\Sfnt{Z}$ and $\Sfnt{S}$. Moreover, we have 
 the following two observations: 
 1) If $Z_{j}=*$ or $Z_{j}\ne 00$, $S_j=0$ because in both cases it is 
 true that $L\leq j$. If $Z_{j}=00$, then $L\geq j$ and so the probability 
 that $S_{j}=0$ is the probability that the $j$'th trial is the spot-checking 
 trial given that $L\geq j$, which is equal to $1/(2^{k}-j+1)$. 
 Therefore, $S_{j}$ is conditionally independent of $\Sfnt{C}_{\leq j}$, 
 $\Sfnt{Z}_{\leq (j-1)}$ and $\Sfnt{S}_{\leq (j-1)}$ given $Z_{j}$. 
 2) The distribution of $C_{j}Z_{j}$ is conditionally independent of 
 $\Sfnt{S}_{<(j-1)}$ given $\Sfnt{C}_{<j}$, $\Sfnt{Z}_{<j}$ and $S_{j-1}$.
 In view of the above two observations and by the chain rule for 
 probability distributions,  we have  
   \begin{align}\label{eq:chained_dist}
     \mu(\Sfnt{C}\Sfnt{Z}\Sfnt{S}) &= 
     \prod_{j=1}^{2^{k}}
     \mu(C_{j}Z_{j} S_{j}|\Sfnt{C}_{<j}\Sfnt{Z}_{<j}\Sfnt{S}_{<j})\notag\\
     &=\prod_{j=1}^{2^{k}} \mu(S_{j}|\Sfnt{C}_{\leq j}\Sfnt{Z}_{\leq j}\Sfnt{S}_{<j})
     \mu(C_{j}Z_{j}|\Sfnt{C}_{<j}\Sfnt{Z}_{<j}\Sfnt{S}_{<j}) \notag \\
     &=
     \prod_{j=1}^{2^{k}}
     \mu(S_{j}|Z_{j})\mu(C_{j}Z_{j}|\Sfnt{C}_{<j}\Sfnt{Z}_{<j}S_{j-1}),
   \end{align}
   where to avoid issues, we define conditional probabilities to be
   $0$ if the condition has probability $0$.   It therefore suffices to specify
   $\mu(C_{j}Z_{j}|\Sfnt{C}_{<j}\Sfnt{Z}_{<j}S_{j-1})$, or
   equivalently, the distributions
   $\mu(C_{j}Z_{j}|\Sfnt{C}_{<j}\Sfnt{Z}_{<j},L\geq j)$ and
   $\mu(C_{j}Z_{j}|\Sfnt{C}_{<j}\Sfnt{Z}_{<j},L<j)$.  The model $\cC$
   is specified by requiring the former to be in $\cM_{j}^{(\text{r})}$ 
   and the latter to be determined with $C_{j}Z_{j}=**$.

In general when constructing chained models for the purpose of randomness 
generation based on PEFs, it is necessary that the next trial's inputs 
satisfy a Markov-chain condition to prevent leaking information about 
the past outputs via the future inputs,
see Sect. IV.A of Ref.~\cite{knill:qc2017a}. Here, the inputs
are chosen independently of the past outputs so the necessary conditions
are automatically satisfied. Formally, it is necessary to include
the condition $L\geq j$ during a block as a part of the inputs for the trial $j$
and  keep calibration information used during an experiment  
 as private information accessible only to the experimentalists, 
as explained in Sect. IV.A and Sect. IV.B of Ref.~\cite{knill:qc2017a}.

As for a single trial, for a block we can define the PEF with positive power 
$\beta$ as a non-negative function $G: \Sfnt{cz} \mapsto G(\Sfnt{cz})$ 
satisfying the PEF inequality 
\begin{equation}\label{eq:pef_def2}
  \sum_{\Sfnt{cz}} \mu(\Sfnt{C}=\Sfnt{c}, \Sfnt{Z}=\Sfnt{z}) G(\Sfnt{cz}) \mu(\Sfnt{C}=\Sfnt{c}|\Sfnt{Z}=\Sfnt{z})^{\beta}\leq 1
\end{equation}
for each probability distribution $\mu(\Sfnt{CZ})$ in the model $\cC$ for the block. 
When the model $\cC$ is obtained by chaining the trial models $\cM_j$ as defined above, 
the PEFs satisfy the following chaining property: If for each possible trial $j$ in 
a block $F_j$ is a PEF with power $\beta$ for $\cM_j$, then the product 
$G=\prod_{j=1}^{2^{k}} F_j$ is a PEF with the 
same power $\beta$ for the chained model $\cC$ (see the proof of Thm.~9 in 
Ref.~\cite{zhang:qc2018a}).  Therefore, we need only to construct trial-wise PEFs 
$F_j$ in order to obtain a PEF for a block.  When we do so, we chose $F_{j}=1$ when 
$j>L$, so that $G=\prod_{j=1}^{L} F_{j}$.

We emphasize that in this work the input distribution $\nu_{j}(XY)$ is assumed to 
be exact as given in Eq.~\eqref{eq:fixed_input_dist}. 
When the random bits used for determining both the actual block length and the settings 
choices in the spot-checking trial are not perfectly uniform but have some 
adversarial biases, the input distribution $\nu_{j}(XY)$ can deviate from Eq.~\eqref{eq:fixed_input_dist}. 
We may then constrain it in a convex polytope by generalizing the arguments presented in 
Refs.~\cite{knill:qc2017a,zhang_y:qc2018a}. In this case, we can still certify
randomness by probability estimation but at the cost of a larger number of trials
to certify a fixed amount of randomness as compared with the case of no adversarial 
bias, an issue which deserves further investigation. 
Allowing for adversarial biases in the input bits for spot checks and settings 
choices is required but not sufficient for randomness amplification~\cite{acin:qc2016a, Bera2017}.  
Randomness amplification with protocols such as ours also requires accounting for adversarial 
biases in the seed bits used by the extractor. There are extractors that can handle such biases 
and can be used for randomness amplification~\cite{kessler:qc2017a}.

\subsection{Certifying smooth min-entropy}
\label{subsect:smooth_entropy}
Suppose that the number of blocks actually observed in an experiment is $n_b$.
 Denote the sequence of inputs of the $i$'th block by $\Sfnt{Z}_i$, and the sequence of the
corresponding outputs by $\Sfnt{C}_i$.  Furthermore, let
$\vec{\Sfnt{Z}}$ be the sequence of inputs of the whole experiment,
and $\vec{\Sfnt{C}}$ be the sequence of the corresponding
outputs. That is, $\vec{\Sfnt{Z}}=(\Sfnt{Z}_1,...,\Sfnt{Z}_{n_b})$ and
$\vec{\Sfnt{C}}=(\Sfnt{C}_1,...,\Sfnt{C}_{n_b})$. For each block
indexed by $i$ we can construct its model $\cC_i$ according to 
Sect.~\ref{subsect:pefs}. Each model $\cC_i$ is the same as the model $\cC$ constructed in
Sect.~\ref{subsect:pefs} but with the random variables associated with
the $i$'th block.  Then, in the same way as we constructed the chained
model for a block, we can construct the model $\cH$ for the whole
experiment by chaining the models $\cC_i$, $i\in \{1,...,n_b\}$. The
previously mentioned Markov-chain conditions are again satisfied
because the inputs for a block are chosen according to a known 
distribution independent of the past.

We would like to certify the amount of extractable randomness in the outputs 
$\vec{\Sfnt{C}}$ conditional on the inputs $\vec{\Sfnt{Z}}$ and the classical 
side information $E$. 
To quantify the amount of extractable randomness,  we define the following 
quantities: 1) Given the joint distribution $\mu$ of the inputs $\vec{\Sfnt{Z}}$, outputs 
$\vec{\Sfnt{C}}$ and classical side information $E$, the quantity $\sum_{\vec{\Sfnt{z}}e}\mu(\vec{\Sfnt{z}}e) \max_{\vec{\Sfnt{c}}}(\mu(\vec{\Sfnt{c}}|\vec{\Sfnt{z}}e))$ is called the (average) maximum guessing probability $P_{\text{guess}}(\vec{\Sfnt{C}}|\vec{\Sfnt{Z}}E)_{\mu}$ of $\vec{\Sfnt{C}}$ given $\vec{\Sfnt{Z}}$ and $E$ according to 
$\mu$, and the quantity $-\log_2(P_{\text{guess}}(\vec{\Sfnt{C}}|\vec{\Sfnt{Z}}E)_{\mu})$ is called the (classical) 
$\vec{\Sfnt{Z}}E$-conditional min-entropy $H_{\min}(\vec{\Sfnt{C}}|\vec{\Sfnt{Z}}E)_{\mu}$ of  $\vec{\Sfnt{C}}$ according 
to $\mu$. For discussions of these quantities, see Refs.~\cite{koenig:qc2008a, koenig:qc2009a}.  
2) The total-variation distance between two distributions $\mu$ and $\nu$ of $X$ is defined as  
 \begin{equation}
  \TV(\mu,\nu) = \frac{1}{2}\sum_{x}|\mu(x)-\nu(x)|.
  \label{eq:def_tv}
 \end{equation}
3) The distribution $\mu$ of $\vec{\Sfnt{C}}\vec{\Sfnt{Z}}E$ has $\epsilon_s$-smooth 
maximum guessing probability $p$ if there exists a distribution $\nu$ of $\vec{\Sfnt{C}}\vec{\Sfnt{Z}}E$ 
with $\TV(\nu,\mu)\leq\epsilon_s$ and $P_{\text{guess}}(\vec{\Sfnt{C}}|\vec{\Sfnt{Z}}E)_{\nu} \leq p$. 
The minimum $p$ for which $\mu$ has $\epsilon_s$-smooth maximum guessing probability $p$ 
is denoted by $P^{\epsilon_s}_{\text{guess}}(\vec{\Sfnt{C}}|\vec{\Sfnt{Z}}E)_{\mu}$.  The quantity
$H_{\min}^{\epsilon_s}(\vec{\Sfnt{C}}|\vec{\Sfnt{Z}}E)_{\mu}=-\log_{2} \big(P^{\epsilon_s}_{\text{guess}}(\vec{\Sfnt{C}}|\vec{\Sfnt{Z}}E)_{\mu}\big)$ 
is called the (classical) $\epsilon_s$-smooth $\vec{\Sfnt{Z}}E$-conditional min-entropy of $\vec{\Sfnt{C}}$ 
according to $\mu$.  This quantity is a specialization of the quantum smooth conditional 
min-entropy~\cite{renner:qc2006a} to probability distributions.

We consider an arbitrary joint distribution $\mu$ of $\vec{\Sfnt{Z}}$, $\vec{\Sfnt{C}}$ and $E$ that satisfies the model 
$\cH$ of the experiment. We say that the distribution $\mu(\vec{\Sfnt{C}}\vec{\Sfnt{Z}}E)$ satisfies the model $\cH$ if 
for each value $e$ of $E$, the conditional distribution $\mu(\vec{\Sfnt{C}}\vec{\Sfnt{Z}}|e)$, viewed as a distribution 
of $\vec{\Sfnt{C}}$ and $\vec{\Sfnt{Z}}$, is in the model $\cH$. Given an arbitrary distribution $\mu(\vec{\Sfnt{C}}\vec{\Sfnt{Z}}E)$ satisfying the model $\cH$, we quantify the amount of extractable randomness by the smooth conditional min-entropy $H_{\min}^{\epsilon_s}(\vec{\Sfnt{C}}|\vec{\Sfnt{Z}}E)_{\mu}$ defined as above. 
Our goal is to obtain a lower bound on $H_{\min}^{\epsilon_s}(\vec{\Sfnt{C}}|\vec{\Sfnt{Z}}E)_{\mu}$ without knowing which particular distribution $\mu(\vec{\Sfnt{C}} \vec{\Sfnt{Z}} E)$ describes the experiment. 
For this,  let the PEF with power $\beta$ for the $i$'th block be $G_{i}$, which is a function of $\Sfnt{C}_i$ and $\Sfnt{Z}_i$, and let $T$ be the product of block-wise PEFs, that is, $T=\prod_{i=1}^{n_b}G_{i}$. Denote the 
number of possible outputs for a real trial by $|C|$, and let $|\vec{\Sfnt{C}}|= |C|^{n_b\times 2^k}$ be the maximum number of possible outputs after $n_b$ blocks, where each block has at most $2^k$ real trials.  A lower bound on  
the smooth conditional min-entropy is obtained according to Thm.~1 of Ref.~\cite{zhang:qc2018a} with the replacement 
of the trial-wise PEFs by the block-wise PEFs.  We state the theorem for our case of interest as follows: 
\begin{theorem} \label{thm:classical_smooth_min_entropy_bound} \emph{(Thm.~1 and Thm.~11 of Ref.~\cite{zhang:qc2018a})} 
Let $1\geq \epsilon_s>0$ and $1\geq  p \geq 1/|\vec{\Sfnt{C}}|$.  Define $\Phi$ to be the event that  
$T \geq 1/(p^\beta\epsilon_s)$.  Suppose that $\mu$ is an arbitrary joint probability distribution of the inputs 
$\vec{\Sfnt{Z}}$, the outputs $\vec{\Sfnt{C}}$ and the classical side information $E$ satisfying the model $\cH$. 
Let $\kappa=\Prob_{\mu}(\Phi)$ be the probability of the event $\Phi$ according to $\mu$, and denote the 
distribution of $\vec{\Sfnt{Z}}$, $\vec{\Sfnt{C}}$ and $E$ conditional on $\Phi$ by $\mu_{|\Phi}$. Then the smooth conditional min-entropy given $\Phi$ satisfies 
\begin{equation}
    H_{\min}^{\epsilon_s}(\vec{\Sfnt{C}}|\vec{\Sfnt{Z}}E)_{\mu_{|\Phi}} \geq  -\log_2(p)+ \frac{1+\beta}{\beta}\log_2(\kappa). 
    \label{eq:classical_smooth_min_entropy_bound0}
\end{equation}
\end{theorem} 
We remark that the bound in Eq.~\eqref{eq:classical_smooth_min_entropy_bound0} implies the following statement: 
For every $\kappa'>0$ and for each probability distribution $\mu(\vec{\Sfnt{C}}\vec{\Sfnt{Z}}E)$ 
satisfying the model $\cH$, \emph{either} the probability of the event $\Phi$ according to $\mu$ is less than 
$\kappa'$ \emph{or} the smooth conditional min-entropy given $\Phi$ satisfies
  \begin{equation}
    H_{\min}^{\epsilon_s}(\vec{\Sfnt{C}}|\vec{\Sfnt{Z}}E)_{\mu_{|\Phi}} \geq  -\log_2(p)+ \frac{1+\beta}{\beta}\log_2(\kappa'). 
    \label{eq:classical_smooth_min_entropy_bound}
  \end{equation}
The event $\Phi$ can
be interpreted as the event that the experiment succeeds. When the
experiment succeeds, we compose the smooth conditional min-entropy
bound in Eq.~\eqref{eq:classical_smooth_min_entropy_bound0} with a
classical-proof strong extractor to obtain near-uniform random bits
 as detailed in the next section.

\section{Randomness generation}
\label{sect:rand_gen}
Our goal is to design a sound randomness-generation protocol, meaning
that the protocol has guaranteed performance no matter how low the
success probability is.  In Sect.~\ref{subsec:extractor} we introduce
the extractor used, which determines the choices of various parameters
in the protocol. Then in Sect.~\ref{subsec:soundness} we formalize the
definition of soundness. Finally in Sect.~\ref{subsec:protocol}, we
present our randomness-generation protocol and prove its soundness.

\subsection{Classical-proof strong extractors}
\label{subsec:extractor}
 An extractor is a function $\cE:(C,S)\mapsto R$, which extracts near-uniform random
 bits from the input randomness source $C$ with the help of the seed $S$ and stores 
 the extracted bits in $R$.  In this work, we assume that the seed $S$ is uniformly 
 random and independent of all other random variables.   
 Suppose that before applying the extractor the joint distribution of the input $C$, 
 the seed $S$ and the classical side information $E$ is given by the product of 
 distributions $\mu(CE)$ and $\tau(S)$, where $\mu(CE)$ denotes the distribution 
 of $C$ and $E$, and $\tau(S)$ is the uniform distribution of $S$.  
 After applying the extractor, the joint distribution of the output $R$, the seed $S$ 
 and the side information $E$ is denoted by $\pi_{\mu}(RSE)$, where the subscript $\mu$
 indicates that the distribution $\pi$ is obtained by applying the extractor with 
 the distribution $\mu$. The goal is for the distribution $\pi_{\mu}(RSE)$ to be close to 
 the product of distributions $\tau(RS)$ and $\mu(E)$, where $\tau(RS)$ is the uniform 
 distribution of $R$ and $S$ together, and $\mu(E)$ is the marginal distribution 
 of $E$ according to $\mu(CE)$. Let $|C|$, $|R|$ and $|S|$ denote the numbers of 
 possible values taken by $C$, $R$ and $S$, respectively, and define $m_{\text{in}}=\log_{2}(|C|)$, 
 $k_{\text{out}}=\log_{2}(|R|)$ and $d_{\text{s}}=\log_{2}(|S|)$. 
 Then, a function $\cE:(C,S)\mapsto R$ is called a classical-proof strong extractor 
 with parameters $(m_{\text{in}},d_{\text{s}},k_{\text{out}},\sigma_{\text{in}},\epsx)$ 
 if for every distribution $\mu(CE)$ with  
 conditional min-entropy $H_{\min}(C|E)_{\mu}\geq \sigma_{\text{in}}$ bits, the joint distribution 
 of $R$ and $S$ is close to uniform and independent of the classical side information $E$
 in the sense that the total-variation distance $\TV\big(\pi_{\mu}(RSE), \tau(RS)\mu(E)\big) \leq \epsx$. 

To ensure the proper functioning of an extractor $\cE$, the parameters $(m_{\text{in}},d_{\text{s}},k_{\text{out}},\sigma_{\text{in}},\epsx)$ 
need to satisfy a set of constraints, called \emph{extractor constraints}. The extractor constraints 
depend on the specific classical-proof strong extractor used, but these constraints 
always include that $1\leq \sigma_{\text{in}} \leq m_{\text{in}}$, $d_{\text{s}}\geq 0$, $k_{\text{out}}\leq \sigma_{\text{in}}$, and $0<\epsx\leq 1$.
In this work, we use Trevisan's extractor~\cite{trevisan:qc2001a} as implemented by Mauerer, 
Portmann and Scholz~\cite{mauerer:qc2012a}, which is a classical-proof strong extractor requiring a 
short seed. We refer to this extractor as the TMPS extractor $\cE_{\TMPS}$. To apply the TMPS extractor, 
additional extractor constraints~\cite{bierhorst:qc2018a} are 
\begin{align} \label{eq:tmps_con}
& k_{\text{out}}+4\log_2 (k_{\text{out}}) \le \sigma_{\text{in}}-6 +  4\log_2(\epsx), \notag \\
& d_{\text{s}} \le w^2 \max \left(2, 1+
    \left\lceil\frac{\log_2(k_{\text{out}}-e)-\log_2(w-e)}{\log_2(e)-\log_2(e-1)}\right\rceil\right), 
\end{align}
where $w$ is the smallest prime larger than $2\lceil\log_2(4m_{\text{in}}k_{\text{out}}^2/\epsx^2)\rceil$. 
We remark that when the extractor error $\epsx$ is specified in trace distance, Trevisan's extractor 
with the above constraints actually works in the presence of quantum side information, 
as shown in Refs.~\cite{de:2009, mauerer:qc2012a}.

\subsection{Protocol soundness}
\label{subsec:soundness}

Consider a generic randomness-generation protocol $\cG$, where there is a binary 
flag $P$ whose value $0$ or $1$ indicates failure or success, respectively. 
Conditional on the success event $P=1$, the protocol $\cG$ produces not 
only a string of fresh random bits stored in $R$, but also a string of 
previously used random bits stored in $S'$. The bit string $R$ is of 
length $k_{\text{out}}$, and the bit string $S'$ is of length $d_{\text{s}}$ consisting 
of the random seed $S$ used for randomness extraction. The protocol outputs 
$R$, $S'$ and $P$ are determined not only by the results of 
the considered experiment, but also by the specific classical-proof strong 
extractor used and its seed $S$. 
  
Recall that the model for the experiment considered with inputs $\vec{\Sfnt{Z}}$
and outputs $\vec{\Sfnt{C}}$ is $\cH$. Consider an arbitrary joint distribution $\mu$ 
of the inputs $\vec{\Sfnt{Z}}$, the outputs $\vec{\Sfnt{C}}$ and the classical side
information $E$ satisfying the model $\cH$. Suppose that $\mu$ is the relevant distribution
before running the protocol. Let $\pi_{\mu}$ be the distribution of the protocol outputs 
$R,S',P$, the experiment inputs $\vec{\Sfnt{Z}}$ and the side information $E$ after 
running the protocol, where the subscript $\mu$ of $\pi$ indicates that the 
distribution $\pi$ is induced by $\mu$. The distribution conditional on the success 
event $P=1$ is $\pi_{\mu}(RS'\vec{\Sfnt{Z}}E|P=1)$. A randomness-generation protocol 
$\cG$ is \emph{$\epsilon$-sound} for the distribution 
$\mu(\vec{\Sfnt{C}}\vec{\Sfnt{Z}}E)$ if there exists a distribution 
$\nu(\vec{\Sfnt{Z}}E)$ such that
\begin{equation}
  \TV\big(\pi_{\mu}(RS'\vec{\Sfnt{Z}}E|P=1), \tau(RS')\nu(\vec{\Sfnt{Z}}E)\big)
  \Prob_{\mu}(P=1)\leq\epsilon,\label{eq:soundness_def}
  \end{equation}
  where $\tau(RS')$ is the uniform distribution over all possible values of the variables $R$ and 
  $S'$ together,  and $\Prob_{\mu}(P=1)$ is the probability of success according to the distribution 
  $\mu(\vec{\Sfnt{C}}\vec{\Sfnt{Z}}E)$.  Our goal is to obtain an 
  $\epsilon$-sound protocol for the model $\cH$ of our experiment, that is, a protocol which is 
  $\epsilon$-sound for all distributions $\mu(\vec{\Sfnt{C}}\vec{\Sfnt{Z}}E)$ satisfying the model $\cH$. 
  Note that it may be desirable to have the total-variation distance conditional on success
   be bounded from above by $\delta$ given that the success probability is larger
   than some threshold $\kappa$. For this it suffices to choose the soundness error $\epsilon$ as 
   $\epsilon\leq \delta\kappa$.  If one wishes to be equally conservative for both $\delta$ and $\kappa$, 
   it makes sense to set $\epsilon=\delta^{2}$. 

  We remark that a protocol $\cG$ is called \emph{$\eta$-complete} for a model $\cH$ 
  if there exist a distribution $\nu(\vec{\Sfnt{C}}\vec{\Sfnt{Z}}E)$ satisfying the model 
  such that the success probability according to $\nu$ satisfies $\Prob_{\nu}(P=1)\geq \eta$.
  Completeness is an important factor to consider when designing an experiment, while soundness 
  guarantees the performance of the protocol regardless of the actual implementation of the 
  experiment.

\subsection{PEF-based randomness-generation protocol}
\label{subsec:protocol}
Recall that we consider an experiment with a sequence of time-ordered blocks, where the length of  
each block is uniformly at random chosen from the set $\{1,...,2^k\}$ and each block $i$ has inputs  $\Sfnt{Z}_i$ 
and outputs $\Sfnt{C}_i$. Suppose that in the experiment at most $N_b$ blocks are acquired. 
Then the maximum length in bits of the outputs $\vec{\Sfnt{C}}$ of the whole experiment  
is $m_{\text{in}}=N_b\times 2^k \times k_{c}$, where $k_{c}$ is the length in bits of the
outputs for each real trial in a block. In our experiment, $k_{c}=2$ considering that in each real trial 
there is a binary output for each of Alice and Bob. To extract $k_{\text{out}}$ random bits  at
soundness error $\epsilon$ from the outputs $\vec{\Sfnt{C}}$ with the help of an extractor $\cE$, the lower bound 
$\sigma_{\text{in}}$ on the conditional min-entropy of $\vec{\Sfnt{C}}$ certified using PEFs with power $\beta$, 
the seed length $d_{\text{s}}$, and the extractor error $\epsx$ need to be chosen from the set $\cX$ defined by 
\begin{align}\label{eq:chi_set}
\cX=\{(\sigma_{\text{in}}, d_{\text{s}},\epsx): & 
\textrm{ the parameters $(m_{\text{in}},d_{\text{s}},k_{\text{out}},\sigma_{\text{in}},\epsx<\epsilon)$ 
satisfy the extractor constraints for $\cE$}, \notag \\
& \textrm{and $\sigma_{\text{in}}\leq m_{\text{in}}+\frac{1+\beta}{\beta}\log_2(\epsilon)$}\}. 
\end{align}
The protocol for end-to-end randomness generation is displayed in Protocol~\ref{prot:condgen_direct}, 
where the notation $0^{\conc l}$ denotes a string of $l$ consecutive zeros and the notation
$\Sfnt{c}0^{\conc l}$ denotes a string obtained by padding $\Sfnt{c}$ with $0^{\conc l}$. We emphasize 
that the parameters $(N_b, m_{\text{in}}, \sigma_{\text{in}}, \beta, d_{\text{s}}, \epsx)$   
specified above are determined before running the protocol. In practice, they are determined
using the commissioning data collected before the randomness-generation experiment, see 
Sect.~\ref{sect:commission} for details. We call the difference $\epse=\epsilon-\epsx$ the entropy error. 
In the rest of this subsection, we use $\Sfnt{Z}$ and $\Sfnt{C}$ to denote the sequences of 
inputs and outputs of real trials in a block, and use $\Sfnt{z}$ and $\Sfnt{c}$ to denote
the corresponding sequences of observed inputs and outputs. For each virtual trial in a block,
we set its inputs and outputs to be $0^{\conc k_z}$ and $0^{\conc k_c}$, where $k_{z}$ and $k_{c}$ 
are the lengths in bits of the inputs and outputs for each real trial. 

\RestyleAlgo{boxruled}
\vspace*{\baselineskip}
\begin{algorithm}
  \caption{Input-conditional randomness generation.}\label{prot:condgen_direct}
  
  \Input{
  \begin{itemize}
  \item $k_{\text{out}}$---the number of fresh random bits to be generated. 
  \item $\epsilon$---the soundness error satisfying $\epsilon\in(0,1]$.
  \end{itemize} 
  }

  \Given{
  \begin{itemize}
  \item $N_b$---the maximum number of blocks to be acquired from the experiment. 
  \item $G_i$---the PEF with power $\beta$ for each block $i$. \\
  \item $\cE$---the classical-proof strong extractor used. 
  \end{itemize}  
  }

  \Output{$R$, $S'$, $P$ as specified in the first paragraph of Sect.~\ref{subsec:soundness}.}

   \BlankLine
  
    Choose $(\sigma_{\text{in}}, d_{\text{s}}, \epsx)$ from the set $\cX$ defined in Eq.~\eqref{eq:chi_set}
    \tcp*{Ensure the set $\cX$ to be non-empty.} 
  
    Get an instance $s$ of the uniformly random seed $S$ of length $d_{\text{s}}$\; 
    
    Set $\epse=(\epsilon-\epsx)$, $q=2^{-\sigma_{\text{in}}}\epsilon$, and $t_{\min}=1/(q^{\beta}\epse)$\;

    Set $t=1$\;

    \For{$i\gets1$ \KwTo $N_b$}{
     Run the experiment to acquire a block of real trials with inputs $\Sfnt{z}_i$ and outputs $\Sfnt{c}_i$\; 
     Compute $g_i=G_i(\Sfnt{c}_i\Sfnt{z}_i)$, and update $t$ as $t=t\times g_i$\;
     Set $\Sfnt{c}'_i=\Sfnt{c}_i 0^{\conc k_{c}\times (2^k-l_i)}$, where $l_i$ is the actual length of
     the $i$'th block\; 
     \If{$t\geq t_{\min}$}    
    { Record the number of blocks actually acquired as $n_b=i$, and stop acquiring the future blocks\;
      Set $\Sfnt{c}'_j=0^{\conc k_{c}\times 2^k}$ for $j\in \{(n_b+1),...,N_b\}$, and set $\vec{\Sfnt{c}'}=(\Sfnt{c}'_1,...,\Sfnt{c}'_{N_b})$\;
      Return $P=1$, $R=\cE(\vec{\Sfnt{c}'},s)$, $S'=s$   
      \tcp*{Protocol succeeded.} 
    }     
    }
    \If{$t<t_{\min}$}    
    { 
      Record the number of blocks actually acquired as $n_b=N_b$\;
      Return $P=0$,  $R=0^{\conc k_{\text{out}}}$, $S'=s$   
      \tcp*{Protocol failed.} 
    }     
\end{algorithm}

Several remarks on the implementation of Protocol~\ref{prot:condgen_direct} are as follows. First, we 
assume that the set $\cX$ defined in Eq.~\eqref{eq:chi_set} is non-empty. This assumption needs to be 
checked before invoking the protocol, and the input parameters can be adjusted to ensure that the 
assumption holds. Second, to extract uniform random bits from the outputs of an experiment, all 
extractors studied in literature require the length of the experimental outputs to be fixed beforehand.
However, in our experiment the length $L$ of a block is not prefixed but uniformly 
at random chosen from the set $\{1,...,2^k\}$. So, we pad the outputs $\Sfnt{c}$ actually observed in 
a block with $0^{\conc k_{c}\times (2^k-l)}$, a string of $k_{c}\times (2^k-l)$ consecutive zeros 
with $l$ being the actual block length, to ensure the fixed-length requirement.  
Third, since the constant function $F=1$ is a valid 
PEF with any positive power for a trial~\cite{zhang:qc2018a, knill:qc2017a}, the constant function $G=1$ 
is a valid PEF with any positive power for a block of an arbitrary length. Therefore, if the 
parameter $t$ in Protocol~\ref{prot:condgen_direct} takes a value larger than $t_{\min}$ after 
the $i$'th block where $i<N_b$, we can set the PEFs for all future blocks to be $G=1$ such 
that the results of the future blocks will not affect the value of $t$. Equivalently, we
can stop acquiring the future blocks and set the outputs of each future block to be 
$0^{\conc k_{c} \times 2^k}$ (in order to ensure the fixed-length requirement by the extractor). 
We call the blocks that are not actually acquired in an experiment \emph{virtual} blocks.
Fourth, for each block $i$ which is actually acquired, we construct its PEF with power $\beta$ as 
$G_i=\prod_{j=1}^{L_i}F_{ij}$, where $F_{ij}$ is a PEF with power $\beta$ for the $j$'th real trial 
in the $i$'th block and $L_i$ is the block length. Before acquiring the $j$'th trial in the $i$'th block, 
the PEF $F_{ij}$ for this trial should be fixed. Otherwise, the soundness of the protocol is not assured. 

The soundness of Protocol~\ref{prot:condgen_direct} can be proven by composing 
Thm.~\ref{thm:classical_smooth_min_entropy_bound} with the classical-proof strong extractor 
$\cE$ used in the protocol. 

\begin{theorem}\label{thm:condgen_direct}
  Protocol~\ref{prot:condgen_direct} is an $\epsilon$-sound randomness-generation protocol for 
  the model $\cH$ of the experiment. 
\end{theorem}

\begin{proof}
  Let the number of blocks actually acquired in the experiment be $n_b\leq N_b$, and 
  let $\mu(\vec{\Sfnt{C}} \vec{\Sfnt{Z}} E)$ 
  be an arbitrary distribution satisfying the model $\cH$ of the experiment. Here 
  $\vec{\Sfnt{C}}=(\Sfnt{C}_1,...,\Sfnt{C}_{n_b})$ and $\vec{\Sfnt{Z}}=(\Sfnt{Z}_1,...,\Sfnt{Z}_{n_b})$. 
  Suppose that when running the protocol, the inputs and outputs of the experiment are instantiated 
  to $\vec{\Sfnt{z}}$ and $\vec{\Sfnt{c}}$ according to $\mu(\vec{\Sfnt{C}} \vec{\Sfnt{Z}} E)$. 
  For each block $i$ actually acquired, let $\Sfnt{C}'_i=\Sfnt{C}_i 0^{\conc k_{c} \times (2^k-l_i)}$ 
  and  $\Sfnt{c}'_i=\Sfnt{c}_i 0^{\conc k_{c} \times (2^k-l_i)}$, where $k_{c}$ is the length in bits of 
  the outputs for each real trial and $l_i$ is the actual length of the $i$'th block. Similarly,
  we have $\Sfnt{Z}'_i=\Sfnt{Z}_i 0^{\conc k_{z} \times (2^k-l_i)}$ and 
  $\Sfnt{z}'_i=\Sfnt{z}_i 0^{\conc k_{z} \times (2^k-l_i)}$, where $k_{z}$ is the length in bits of 
  the inputs for each real trial. Moreover, for each virtual block $j$ with $j\in \{(n_b+1),...,N_b\}$ we set 
  $\Sfnt{C}'_j=\Sfnt{c}'_j=0^{\conc k_{c}\times 2^k}$ and $\Sfnt{Z}'_j=\Sfnt{z}'_j=0^{\conc k_{z}\times 2^k}$. 
  Define $\vec{\Sfnt{C}'}=(\Sfnt{C}'_1,...,\Sfnt{C}'_{N_b})$ and $\vec{\Sfnt{Z}'}=(\Sfnt{Z}'_1,...,\Sfnt{Z}'_{N_b})$. 
  The distribution $\mu(\vec{\Sfnt{C}}' \vec{\Sfnt{Z}'} E)$ is fully determined by the 
  distribution $\mu(\vec{\Sfnt{C}} \vec{\Sfnt{Z}} E)$. We emphasize that for each possible 
  distribution $\mu(\vec{\Sfnt{C}} \vec{\Sfnt{Z}} E)$ satisfying the model $\cH$ of the experiment, 
  we can determine the corresponding distribution $\mu(\vec{\Sfnt{C}}' \vec{\Sfnt{Z}'} E)$.
  The set of all possible such distributions $\mu(\vec{\Sfnt{C}}' \vec{\Sfnt{Z}'} E)$ is defined to 
  be the generalized model $\cH'$ of the experiment. 
  For each block $i\leq n_b$ which is actually acquired in the experiment, as the constant function 
  $F=1$ is a valid PEF with any positive power for a trial~\cite{zhang:qc2018a, knill:qc2017a}, the function
  $G_i(\Sfnt{C}_i\Sfnt{Z}_i)$ is a PEF with power $\beta$ for the block with results $\Sfnt{C}'_i\Sfnt{Z}'_i$. 
  For the same reason, the constant function $G_j=1$ is a PEF with power $\beta$ for each virtual block
  $j$ with $j>n_b$. Therefore, if we set 
  $T(\vec{\Sfnt{C}'} \vec{\Sfnt{Z}'})=\prod_{i=1}^{n_b} G_i(\Sfnt{C}_i\Sfnt{Z}_i)$ and write 
  the event $\Phi=\left\{\vec{\Sfnt{c}'}\vec{\Sfnt{z}'}: 
  T(\vec{\Sfnt{c}'}\vec{\Sfnt{z}'})\geq t_{\min}=1/(q^{\beta}\epse)\right\}$, then 
  when $\vec{\Sfnt{c}'}\vec{\Sfnt{z}'} \in \Phi$, the protocol succeeds, that is, $P=1$.  Let 
  the probability of success according to $\mu$ be $\kappa=\Prob_{\mu}(\Phi)$, and denote 
  the joint distribution of $\vec{\Sfnt{C}'}, \vec{\Sfnt{Z}'}$ and $E$ conditional on success
  by $\mu(\vec{\Sfnt{C}'} \vec{\Sfnt{Z}'} E|\Phi)$, abbreviated as $\mu_{|\Phi}$.
  
  We first consider the case $\kappa\in[\epsilon,1]$. 
  With the substitutions $\cH \rightarrow \cH'$, $\epsilon_s \rightarrow \epse/\kappa$ and 
  $p \rightarrow q\kappa^{1/\beta}$, 
 the event $\Phi$ in the statement of Thm.~\ref{thm:classical_smooth_min_entropy_bound} becomes the same as 
 the event $\Phi$ defined above, and so the parameter $\kappa$ in the statement of 
 Thm.~\ref{thm:classical_smooth_min_entropy_bound} becomes the same as the parameter $\kappa$ introduced above.
 To apply Thm.~\ref{thm:classical_smooth_min_entropy_bound}, we need to check that the parameter $p$ in the 
 statement of Thm.~\ref{thm:classical_smooth_min_entropy_bound} satisfies  
 \begin{align*}
 p\rightarrow q\kappa^{1/\beta}&=2^{-\sigma_{\text{in}}}\epsilon \kappa^{1/\beta}\\
 &\geq 2^{-\sigma_{\text{in}}} \epsilon^{1+1/\beta}\\
 &\geq 1/2^{m_{\text{in}}}=1/|\vec{\Sfnt{C}'}|, 
 \end{align*}
 where the equation in the first line is according to the specification $q=2^{-\sigma_{\text{in}}}\epsilon$ 
 in Protocol~\ref{prot:condgen_direct}, the inequality in the second line follows from 
 $\kappa\geq \epsilon$, and the inequality in the last line follows from   
 $\sigma_{\text{in}}\leq m_{\text{in}}+\frac{1+\beta}{\beta}\log_2(\epsilon)$ 
 as required in the definition of the set $\cX$ (see Eq.~\eqref{eq:chi_set}). 
 Therefore, we can apply Thm.~\ref{thm:classical_smooth_min_entropy_bound} with 
 the above substitutions to obtain  
  \begin{equation}
  H^{\epse/\kappa}_{\min}(\vec{\Sfnt{C}'}|\vec{\Sfnt{Z}'}E)_{\mu_{|\Phi}}\geq -\log_2(q/\kappa). 
    \label{eq:thm:condgen_direct:1}
  \end{equation}
  Considering that $\kappa \geq \epsilon$ and that Protocol~\ref{prot:condgen_direct} sets 
  $q=2^{-\sigma_{\text{in}}}\epsilon$, Eq.~\eqref{eq:thm:condgen_direct:1} implies that 
  \begin{equation}
  H^{\epse/\kappa}_{\min}(\vec{\Sfnt{C}'}|\vec{\Sfnt{Z}'}E)_{\mu_{|\Phi}}\geq \sigma_{\text{in}}. 
    \label{eq:thm:condgen_direct:2}
  \end{equation} 
  According to the definition of smooth conditional min-entropy (see the second paragraph of 
  Sect.~\ref{subsect:smooth_entropy}), there exists a distribution $\nu(\vec{\Sfnt{C}'} \vec{\Sfnt{Z}'} E)$ such that 
  \begin{equation} \label{eq:entropy_witness}
  \TV\big(\mu(\vec{\Sfnt{C}'} \vec{\Sfnt{Z}'} E|\Phi), \nu(\vec{\Sfnt{C}'} \vec{\Sfnt{Z}'} E)\big)\leq \epse/\kappa,
  \end{equation}
  and 
  \begin{equation} 
  H_{\min}(\vec{\Sfnt{C}'}|\vec{\Sfnt{Z}'}E)_{\nu}\geq \sigma_{\text{in}}. 
  \end{equation}
  Because the parameters $(m_{\text{in}},d_{\text{s}},k_{\text{out}},\sigma_{\text{in}},\epsx)$ satisfy the extractor 
  constraints for $\cE$, we can apply the classical-proof strong extractor $\cE$ with the distribution 
  $\nu(\vec{\Sfnt{C}'} \vec{\Sfnt{Z}'} E)$ and the resulting joint distribution $\pi_{\nu}$ of the extractor output $R$, 
  the seed $S$ and the classical  side information $\vec{\Sfnt{Z}'} E$ satisfies
  \begin{equation}\label{eq:ex_witness_1} 
      \TV\big(\pi_{\nu}(R S \vec{\Sfnt{Z}'} E), \tau(R S) \nu(\vec{\Sfnt{Z}'} E)\big)\leq\epsx,   
  \end{equation}
  where $\tau(R S)$ is the uniform distribution of $R$ and $S$,  and $\nu(\vec{\Sfnt{Z}'} E)$ is 
  the marginal distribution of $\vec{\Sfnt{Z}'}$ and $E$ according to $\nu(\vec{\Sfnt{C}'} \vec{\Sfnt{Z}'} E)$. 
  Moreover, since the total-variation distance satisfies the data-processing inequality, 
  Eq.~\eqref{eq:entropy_witness} implies that 
  \begin{equation}\label{eq:ex_witness_2}
  \TV\big(\pi_{\mu}(R S \vec{\Sfnt{Z}'} E|\Phi), \pi_{\nu}(R S \vec{\Sfnt{Z}'}E) \big)\leq \TV\big(\mu(\vec{\Sfnt{C}'} \vec{\Sfnt{Z}'} E|\Phi), \nu(\vec{\Sfnt{C}'} \vec{\Sfnt{Z}'} E)\big) \leq\epse/\kappa.
  \end{equation}
  The triangle inequality for the total-variation distance  together with  Eqs.~\eqref{eq:ex_witness_1} 
  and~\eqref{eq:ex_witness_2} yield
  \begin{equation*}
  \TV\big(\pi_{\mu}(R S \vec{\Sfnt{Z}'} E|\Phi),\tau(R S) \nu(\vec{\Sfnt{Z}'} E)\big)\leq 
  \TV\big(\pi_{\mu}(R S \vec{\Sfnt{Z}'} E|\Phi), \pi_{\nu}(R S \vec{\Sfnt{Z}'}E) \big)+\TV\big(\pi_{\nu}(R S \vec{\Sfnt{Z}'} E), \tau(R S) \nu(\vec{\Sfnt{Z}'} E)\big)\leq \epsx+\epse/\kappa. 
  \end{equation*}
  We multiply both sides by $\kappa$ to obtain the soundness statement 
  \begin{equation*}
  \TV\big(\pi_{\mu}(R S \vec{\Sfnt{Z}'} E|\Phi),\tau(R S) \nu(\vec{\Sfnt{Z}'} E)\big) \kappa  \leq\epsx\kappa +\epse\leq\epsx+\epse=\epsilon. 
  \end{equation*}
  
  For the case $\kappa<\epsilon$, since the total-variation distance cannot be larger than one,
  \begin{equation*}
  \TV\big(\pi_{\mu}(R S \vec{\Sfnt{Z}'} E|\Phi),\tau(R S) \nu'(\vec{\Sfnt{Z}'} E)\big) \kappa \leq \kappa<\epsilon, 
  \end{equation*}
  where $\tau(R S)$ is the uniform distribution as in Eq.~\eqref{eq:ex_witness_1} and 
  $\nu'(\vec{\Sfnt{Z}'} E)$ is an arbitrary distribution of $\vec{\Sfnt{Z}'}$ and $E$.
  
Therefore, the condition for $\epsilon$-soundness is satisfied for  $\mu(\vec{\Sfnt{C}'} \vec{\Sfnt{Z}'} E)$,  independently of the value of the success probability $\kappa$. 
Equivalently, the condition for $\epsilon$-soundness is satisfied for $\mu(\vec{\Sfnt{C}} \vec{\Sfnt{Z}} E)$, 
as the distribution $\mu(\vec{\Sfnt{C}'} \vec{\Sfnt{Z}'} E)$ is fully determined by $\mu(\vec{\Sfnt{C}} \vec{\Sfnt{Z}} E)$. Because $\mu(\vec{\Sfnt{C}} \vec{\Sfnt{Z}} E)$ is an arbitrary distribution satisfying the model $\cH$, 
Protocol~\ref{prot:condgen_direct} is $\epsilon$-sound for the model $\cH$. 
  
\end{proof}

We remark that the above soundness proof, motivated by the soundness proof in the presence of quantum 
side information in Refs.~\cite{knill:qc2018a, zhang_y:qc2020a}, is simpler than that presented in our previous 
work's protocol $\cQ$, Thm. 21 of Ref.~\cite{knill:qc2017a}.  There we designed the protocol and proved 
its soundness using a linear strong extractor rather than a classical-proof strong extractor. By taking 
advantge of the linearity of the extractor, the soundness proof in Ref.~\cite{knill:qc2017a} works with 
the parameter $q=2^{-\sigma_{\text{in}}}$, that is without the additional factor $\epsilon$ as specified 
in Protocol~\ref{prot:condgen_direct}. This has the effect of reducing the success threshold $t_{\min}$ 
for the product of block-wise PEFs by a factor of $\epsilon^{\beta}$. 
Here we choose not to take advantage of this improvement to simplify the soundness proof and presentation of the protocol.

\section{Construction of block-wise PEFs}
\label{sect:pef_constrct}
In this section, we formulate the PEF optimization problem for
randomness generation with block-wise PEFs.  For randomness
generation, the central task is to maximize the number of random bits
generated after a fixed number of blocks, while for randomness
expansion, the goal is different as we care about the difference
between the number of random bits generated and the number of random
bits consumed. The optimization problem for randomness expansion is
based on the optimization problem formulated in this section and will
be detailed in Sect.~\ref{sect:commission}.

\subsection{Formulation of block-wise PEF optimization} 
\label{sect:pef_constrct_thry}

Let $G_{i}(\Sfnt{C}_{i}\Sfnt{Z}_{i})$ be a PEF with power $\beta$ for the $i$'th block in an experiment. 
The experiment successfully implements Protocol~\ref{prot:condgen_direct} if the block-wise PEFs satisfy the 
condition $\prod_{i=1}^{N_b}G_{i}(\Sfnt{C}_{i}\Sfnt{Z}_{i}) \geq t_{\min}= 1/(q^\beta\epse)$ with 
$q=2^{-\sigma_{\text{in}}}\epsilon$,  or equivalently,
 \begin{align}\label{eq:threhd_cond}
 \sum_{i=1}^{N_b}\log_{2}(G_{i}(\Sfnt{C}_{i}\Sfnt{Z}_{i}))/\beta + \log_{2}(\epse)/\beta + \log_{2}(\epsilon) \geq 
 \sigma_{\text{in}}.
 \end{align} 
In this work, we call the quantity $\sum_{i=1}^{N_b}\log_{2}(G_{i}(\Sfnt{C}_{i}\Sfnt{Z}_{i}))/\beta$ the
output entropy witness after $N_b$ blocks.  We call the quantity on the left-hand side of the 
above equation the adjusted output entropy witness, where the adjustment is for both the entropy error $\epse$ 
and the soundness error $\epsilon$. This adjusted output entropy witness is the one plotted in Fig. 4 of the main text.
Hence, for randomness generation we aim to obtain a large expected value of the adjusted output entropy witness  
with as few blocks as possible, supposing that the quantum devices to be used perform as designed. 
The input-conditional output distribution $\nu_{h}(C|Z)$ of these ``honest'' devices is 
thus known, independent and identical for each trial given the inputs, and anticipated to be 
the same as the actual device behavior in the absence of faults or interference. 
In our experiment, the (presumed current) honest-device input-conditional output distribution 
is determined by maximum likelihood from
 calibration data with the method described in Sect.~VIII B of Ref.~\cite{knill:qc2017a}.  
\emph{Before} the experiment,  we can choose values for $\sigma_{\text{in}}$ and $\epse$ (see 
Protocol~\ref{prot:condgen_direct}) and optimize over the block-wise PEFs and the power $\beta$ such that the number  $N_{\textrm{exp},b}$ of blocks required for success with honest devices, as defined below, is minimized.  Then we fix 
the number of blocks $N_b$ available in the experiment to be a number larger than the minimum number of blocks, so 
that the actual experiment succeeds with high probability if the quantum devices used are honest.  Because for 
honest devices, the results of each block are independent and identically 
distributed (i.i.d.), in the pre-experiment optimization we set all $G_{i}$ to be the same. In reality, we update
the block-wise PEFs to be used for future blocks according to incoming calibration data during the experiment. 
The optimization of block-wise PEFs is the same in all cases, except that during the experiment parameters such as 
$\beta$ and $\epse$ are fixed.

Consider a generic next block with results $\Sfnt{C}\Sfnt{Z}$ and model $\cC$. With honest devices
where the input-conditional output distribution for each trial is fixed to be $\nu_{h}(C|Z)$, in principle 
we can determine the distribution $\nu(\Sfnt{C}\Sfnt{Z})$ for the next block's results. However, the explicit expression for $\nu(\Sfnt{C}\Sfnt{Z})$ is somewhat involved. Here we just provide the basics. For determining $\nu(\Sfnt{C}\Sfnt{Z})$, in view of Eq.~\eqref{eq:chained_dist} it suffices to know the conditional distributions $\nu(C_{j}Z_{j}|\Sfnt{C}_{<j}\Sfnt{Z}_{<j}, L\geq j)$ and $\nu(C_{j}Z_{j}|\Sfnt{C}_{<j}\Sfnt{Z}_{<j}, L<j)$, where $L$ is the block-length variable. 
We implicitly assume that the results $\Sfnt{C}_{<j}\Sfnt{Z}_{<j}$ are consistent with whether the condition $L\geq j$ 
holds or not. Given that $L\geq j$, the distribution of $C_{j}Z_{j}$ is independent of 
$\Sfnt{C}_{<j}\Sfnt{Z}_{<j}$ for honest devices. So, 
\begin{align}\label{eq:cond_dist_actual}
\nu(C_{j}Z_{j}|\Sfnt{C}_{<j}\Sfnt{Z}_{<j}, L\geq j)&=\nu(C_{j}Z_{j}|L\geq j) \notag \\
&=\nu(C_{j}|Z_{j},L\geq j)\nu(Z_{j}|L\geq j) \notag \\
&=\nu_{h}(C_{j}|Z_{j})\nu_j(Z_j),
\end{align}
where for the equality in the last line we used the facts that the probability $\nu(C_{j}=c|Z_{j}=z,L\geq j)$ 
is fixed to be $\nu_{h}(C=c|Z=z)$, independent of $j$, and that the input distribution $\nu(Z_{j}|L\geq j)$ 
is given by $\nu_j(Z_j)$ of Eq.~\eqref{eq:fixed_input_dist} with $Z_j=X_jY_j$. On the other hand, when $j>L$ 
the results of the $j$'th trial are fixed to be $Z_j=C_j=*$, independent of the past results 
$\Sfnt{C}_{<j}\Sfnt{Z}_{<j}$, according to model construction. So, 
\begin{equation}\label{eq:cond_dist_filled}
\nu(C_{j}Z_{j}|\Sfnt{C}_{<j}\Sfnt{Z}_{<j}, L<j)=\delta_{C_{j},*}\delta_{Z_{j},*},
\end{equation}
where $\delta$ is the Kronecker delta function.

For honest devices with distribution $\nu(\Sfnt{C}\Sfnt{Z})$ for each block, after $N_b$ blocks  
the adjusted output entropy witness (that is, the left-hand side of Eq.~\eqref{eq:threhd_cond}) has expectation  
$\Exp_{\nu} \big(N_b \log_2(G(\Sfnt{C}\Sfnt{Z}))/\beta+\log_2(\epse)/\beta+ \log_2(\epsilon)\big)$. 
Here, $G(\Sfnt{C}\Sfnt{Z})$ is a block-wise PEF with power $\beta$ for the model $\cC$ for each block.  
Thus, one way to optimize block-wise PEFs is as follows: 
\begin{equation}
  \begin{array}[b]{rl}
    \max_{G}:& 
    \Exp_{\nu} \big(N_b \log_2(G(\Sfnt{C}\Sfnt{Z}))/\beta+\log_2(\epse)/\beta+ \log_2(\epsilon)\big) \\ 
    \textrm{subject to:} & 1) \ G(\Sfnt{c}\Sfnt{z})\geq 0, \textrm{\ for all $\Sfnt{c}\Sfnt{z}$}, \\
                         & 2) \ \sum_{\Sfnt{cz}} \mu(\Sfnt{C}=\Sfnt{c}, \Sfnt{Z}=\Sfnt{z}) G(\Sfnt{cz}) \mu(\Sfnt{C}=\Sfnt{c}|\Sfnt{Z}=\Sfnt{z})^{\beta}\leq 1, \textrm{\ for all $\mu(\Sfnt{CZ})\in \cC$}.  
  \end{array}\label{eq:opt_block_pef}
\end{equation}
Here, the maximum is over all possible block-wise PEFs $G(\Sfnt{C}\Sfnt{Z})$ with power $\beta$ for $\cC$,
not only over the block-wise PEFs constructed by chaining trial-wise PEFs with power $\beta$.
When the block-wise PEFs are constrained to those constructed by chaining trial-wise PEFs and 
when the PEFs for all possible trial positions in a block are constrained to those constructed by a 
well-performing liner-interpolation method as detailed in the next subsection, the block-wise PEF optimization 
is effectively solved.  We emphasize that in the absence of additional constraints on $G$, 
the block-wise PEF returned by the optimization problem of Eq.~\eqref{eq:opt_block_pef} is optimal at $\nu$ 
given the model $\cC$, but every feasible block-wise PEF, for example, the solution returned by the method described in the next subsection, is valid by definition regardless of the actual distributions of the blocks' results 
as long as the possible distributions of these results given the past are in $\cC$.

Before the experiment, we aim to minimize the number of blocks $N_{\textrm{exp},b}$ required for 
successful randomness generation with honest devices.  For this, we define the quantity 
\begin{equation}
  g_b(\beta)=\max_{G}\Exp_{\nu} \big(\log_2(G(\Sfnt{C}\Sfnt{Z}))/\beta\big),
  \label{eq:block_gain_rate}
\end{equation}
where the maximum is over all possible block-wise PEFs $G(\Sfnt{C}\Sfnt{Z})$ with power $\beta$ for $\cC$. 
Thus, for honest devices with distribution $\nu(\Sfnt{C}\Sfnt{Z})$ for each block,  the expectation of 
$\sum_{i=1}^{N_b}\log_{2}(G_{i}(\Sfnt{C}_{i}\Sfnt{Z}_{i}))/\beta$  can be as high as $N_b g_b(\beta)$.  
Success requires that $\sum_{i=1}^{N_b}\log_{2}(G_{i}(\Sfnt{C}_{i}\Sfnt{Z}_{i}))\geq \log_{2}(t_{\min})$, 
where $t_{\min}= 1/(q^\beta\epse)$ with $q=2^{-\sigma_{\text{in}}}\epsilon$. For adequate probability of 
success we therefore need $\log_{2}(t_{\min})/\beta$ smaller than $N_b g_b(\beta)$ by an amount of order 
$\sqrt{N_{b}}$ that is asymptotically small compared to $N_b g_b(\beta)$.  For the present analysis, we 
simply define the minimum number of blocks $N_{\mathrm{exp},b}(\beta)$ required using block-wise PEFs 
with power $\beta$ by the identity 
 \begin{equation}\label{eq:critical_cond}
   N_{\mathrm{exp},b}(\beta) g_b(\beta)= \log_{2}(t_{\min})/\beta
   = \sigma_{\text{in}}-\log_{2}(\epse)/\beta-\log_{2}(\epsilon).
 \end{equation}
Minimizing $N_{\mathrm{exp},b}(\beta)$ over $\beta>0$ gives the minimum number
of blocks $N_{\mathrm{exp},b}$ required according to this simplification:
\begin{equation}
N_{\mathrm{exp},b} = \inf_{\beta>0}N_{\mathrm{exp},b}(\beta).
\label{eq:block_bound}
\end{equation}
The lower bound $N_{\textrm{exp},b}$ may be considered tight to lowest order in the sense that 
one needs only to increase the number of blocks used in practice by an amount that is asymptotically 
small compared to $N_{\textrm{exp},b}$, in order to guarantee sufficient probability of success. 
The probability of success can be estimated according to the distribution of a sum of the i.i.d. 
random variables $\log_{2}(G_{i})$, where $G_i$ is the PEF used for the $i$'th block.  We compute 
the mean and variance of $\log_{2}(G_{i})$ in the next subsection and estimate the probability of 
success accordingly in Sect.~\ref{subsect:para}. In the next subsection, we will provide an effective 
method for determining a lower bound of $g_b(\beta)$. Accordingly, we can obtain an upper bound of 
$N_{\mathrm{exp},b}$.

We remark that for each fixed $\beta$, the quantity $g_b(\beta)$ defined in Eq.~\eqref{eq:block_gain_rate} 
can be identified as the maximum asymptotic randomness rate per block 
witnessed by block-wise PEFs with power $\beta$ when each block has the 
same distribution $\nu(\Sfnt{CZ})$ and is described by the same model $\cC$. Therefore,
the maximum asymptotic randomness rate per block witnessed by all possible block-wise PEFs for 
$\cC$ is $g_{b,\max}=\sup_{\beta>0} g_b(\beta)$.  The justification of the above identification 
is the same as that for the case of trial-wise PEFs as detailed in Refs.~\cite{zhang:qc2018a, knill:qc2017a}.

\subsection{Simplified and effective block-wise PEF optimization}
\label{sect:pef_constrct_impl}

  Consider a generic block with the $j$'th trial model $\cM_{j}$ for
  $j\leq L$.  Given PEFs $F_{j}$ with power $\beta$ for trial models $\cM_{j}$,
  by PEF chaining the product $G=\prod_{j=1}^{L} F_{j}$ is a PEF with power $\beta$
  for the model $\cC$ for the block. In view of the optimization problems formulated 
  in the previous subsection,  we wish to choose $F_{j}$ 
  so as to  maximize the expectation of $\log_{2}(G(\Sfnt{C}\Sfnt{Z})/\beta)$
  for honest devices with distribution $\nu(\Sfnt{C}\Sfnt{Z})$. Instead of optimizing
  $F_{j}$ for each $j$, here we develop a well-performing linear-interpolation method
  to construct the PEFs for all possible trial positions in a block given the 
  optimized PEFs for only three trial positions. We find that the trial-wise PEFs 
  thus obtained perform almost as well as the optimized PEFs for each trial position.

  We begin by computing the quantity 
  $g_{b}(G, \beta)=\Exp_{\nu}\big(\log_{2}(G(\Sfnt{C}\Sfnt{Z})/\beta)\big)$. This  
  quantity can be interpreted as the asymptotic randomness rate witnessed by the block-wise
  PEF $G(\Sfnt{C}\Sfnt{Z})$ with power $\beta$ when  each block has the distribution 
  $\nu(\Sfnt{CZ})$.  Since we choose $F_{j}(C_{j}Z_{j})=1$ for $j> L$, or equivalently, 
  for $C_{j}=Z_{j}=*$, $\log_{2}(F_{j})=0$ for $j>L$ and so 
  \begin{align} \label{eq:gain_exp1}
    \Exp_{\nu}\big(\log_{2}(G(\Sfnt{C}\Sfnt{Z}))\big)
    &= \Exp_{\nu}\Big(\sum_{j=1}^{L}\log_{2}(F_{j}(C_{j}Z_{j}))\Big)\notag\\
    &= \Exp_{\nu}\Big(\sum_{j=1}^{2^{k}}\log_{2}(F_{j}(C_{j}Z_{j}))\Big)\notag\\
    &= \sum_{j=1}^{2^{k}}\Exp_{\nu}\big(\log_{2}(F_{j}(C_{j}Z_{j}))\big).
 \end{align} 
 In view of the law of total expectation and considering that the distribution of 
 the block-length variable $L$ is determined by $\nu(\Sfnt{CZ})$, we can continue 
 where we left off to get 
 \begin{align} \label{eq:gain_exp2}
 \Exp_{\nu}\big(\log_{2}(G(\Sfnt{C}\Sfnt{Z}))\big)      
    &= \sum_{j=1}^{2^{k}}\Big(\Prob_{\nu}(j\leq L)\Exp_{\nu}\big(\log_{2}(F_{j}(C_{j}Z_{j}))|j\leq L\big)+
       \Prob_{\nu}(j>L)\Exp_{\nu}\big(\log_{2}(F_{j}(C_{j}Z_{j}))|j> L\big)\Big)
      \notag\\
    &= \sum_{j=1}^{2^{k}}\Prob_{\nu}(j\leq L)\Exp_{\nu_{j}}\big(\log_{2}(F_{j}(C_{j}Z_{j}))\big)\notag\\
    &= \sum_{j=1}^{2^{k}}\omega_j \Exp_{\nu_{j}}\big(\log_{2}(F_{j}(C_{j}Z_{j}))\big), 
  \end{align}
  where $\omega_j=(2^{k}-j+1)/2^{k}$ is the probability that the block-length variable $L$ 
  is larger than or equal to $j$, and the distribution $\nu_{j}$ in the subscript is the 
  abbreviation of the distribution $\nu(C_{j}Z_{j}|\Sfnt{C}_{<j}\Sfnt{Z}_{<j}, L\geq j)$,
  which is given by Eq.~\eqref{eq:cond_dist_actual} for honest devices. Therefore, we have 
  \begin{equation}
 g_{b}(G, \beta)=\sum_{j=1}^{2^k} \omega_j g_{b,j}(F_j, \beta),  
 \label{eq:block_gain_rate2}
\end{equation}
where $g_{b,j}(F_j, \beta)=\Exp_{\nu_{j}} \big(\log_2(F_j(C_jZ_j))/\beta\big)$.     
Hence, to maximize $g_{b}(G, \beta)$ over trial-wise PEFs, it suffices to maximize
  each $g_{b,j}(F_j, \beta)$ independently over the PEFs  $F_j$ for the trial model 
  $\cM_{j}$ constructed under the condition $L\geq j$. 
  We denote the maximum of $g_{b}(G, \beta)$ over $G=\prod_{j=1}^{L} F_{j}$
  by $g'_{b}(\beta)$ and the maximum of $g_{b,j}(F_j, \beta)$ over $F_j$ by 
  $g_{b,j}(\beta)$.  Then, according to Eq.~\eqref{eq:block_gain_rate2} we have 
  $g'_{b}(\beta)=\sum_{j=1}^{2^k} \omega_j g_{b,j}(\beta)$. 
  Moreover, $g'_{b}(\beta)$ is a lower bound of 
  the quantity $g_b(\beta)$ defined in Eq.~\eqref{eq:block_gain_rate}.
  That is, we have
  \begin{equation}\label{eq:block_gain_rate3}
  g_b(\beta)\geq g'_{b}(\beta)=\sum_{j=1}^{2^k} \omega_j g_{b,j}(\beta).
  \end{equation}

  For the purpose of estimating the probability of success in our protocol 
  implementation,  we need the variance of $\log_{2}(G(\Sfnt{C}\Sfnt{Z}))$ 
  with respect to $\nu(\Sfnt{C}\Sfnt{Z})$.  Although the variance is not needed 
  for optimizing trial-wise PEFs, we obtain an expression for it here which 
  will be used in Sect.~\ref{subsect:para}.  In the same way as deriving
  Eq.~\eqref{eq:gain_exp1}, we have    
    \begin{align}
      \Exp_{\nu}\big(\log_{2}^{2}(G(\Sfnt{C}\Sfnt{Z}))\big) &=
        \Exp_{\nu}\left(\sum_{i=1}^{2^{k}}\sum_{j=1}^{2^{k}}
             \log_{2}(F_{i}(C_{i}Z_{i}))\log_{2}(F_{j}(C_{j}Z_{j}))
           \right)\notag\\
         &=\sum_{i=1}^{2^{k}}\sum_{j=1}^{2^{k}}
\Exp_{\nu}\big(\log_{2}(F_{i}(C_{i}Z_{i}))\log_{2}(F_{j}(C_{j}Z_{j}))\big).
    \end{align}
    Consider the case $i\leq j$. The case $i>j$ can be computed by
    exchanging $i$ and $j$.  By splitting the expression conditional
    on $L< i$, $i\leq L< j$ and $j\leq L$ and taking into account that
    the product $\log_{2}(F_{i}(C_{i}Z_{i}))\log_{2}(F_{j}(C_{j}Z_{j}))=0$ 
    for the first two cases, we get
    \begin{align}
      \Exp_{\nu}\big(\log_{2}(F_{i}(C_{i}Z_{i}))\log_{2}(F_{j}(C_{j}Z_{j}))\big)
      &= \Prob_{\nu}(j\leq L)\Exp_{\nu}\big(\log_{2}(F_{i}(C_{i}Z_{i}))\log_{2}(F_{j}(C_{j}Z_{j}))|j\leq L\big).
    \end{align}
    If $i=j\leq L$,
    \begin{align}
      \Exp_{\nu}\big(\log_{2}(F_{i}(C_{i}Z_{i}))\log_{2}(F_{j}(C_{j}Z_{j}))|j\leq L\big) &= \Exp_{\nu_{j}}\big(
      \log_{2}^{2}(F_{j}(C_{j}Z_{j}))\big).
    \end{align}
    If $i<j\leq L$, considering that in this case the inputs $Z_i$ must be equal to $00$ and that with honest 
    devices the output distributions for the trials $i, j$ are independent conditionally on their respective inputs, 
    we have 
    \begin{equation}
      \Exp_{\nu}\big(\log_{2}(F_{i}(C_{i}Z_{i}))\log_{2}(F_{j}(C_{j}Z_{j}))|j\leq L\big) = \Exp_{\nu_{i}|Z_{i}=00}\big(\log_{2}(F_{i}(C_{i}Z_{i}))\big)
           \Exp_{\nu_{j}}\big(\log_{2}(F_{j}(C_{j}Z_{j}))\big).
    \end{equation}
    In view of the above four equations as well as Eq.~\eqref{eq:gain_exp2}, we can 
    compute the variance of $\log_{2}(G(\Sfnt{C}\Sfnt{Z}))$.

Next, we observe that both the optimization problem of Eq.~\eqref{eq:opt_block_pef} for updating block-wise PEFs during 
an experiment and the optimization problem of Eq.~\eqref{eq:block_bound} for minimizing the number of blocks required for 
success before the experiment are based on determining the quantity $g_b(\beta)$. 
Here we focus on an effective method for lower-bounding $g_b(\beta)$. For this, we first use 
Eq.~\eqref{eq:block_gain_rate3} to reduce the problem of finding $g'_b(\beta)$  
to the problem of finding $g_{b,j}(\beta)$ for each possible trial position $j$ in the block. 
For each $j$, the optimization problem can be formulated as a sequential quadratic program and so can be effectively solved (see Refs.~\cite{zhang:qc2018a, knill:qc2017a} for details). However, a block usually consists of a large number of trials. So, the individual optimizations for all trial positions in a block still take much time.  To save time, it is better to solve the optimization problems for only a few trial positions, and then construct \emph{valid} but maybe \emph{suboptimal} PEFs for the other trial positions by an efficient method. In this way, we obtain a lower bound of $g'_b(\beta)$. For this, we take advantage of the similarity among the trial models $\cM_j$ constructed under the 
condition $L\geq j$, particularly the following proposition: 
\begin{proposition}\label{obs1}
Let $F_j(C_jZ_j)$ be a PEF with power $\beta$ for the trial model $\cM_j$ constructed under the condition $L\geq j$,
where the distribution of inputs $Z_j=X_jY_j$ according to $\cM_j$ is fixed to be $\nu_j(X_jY_j)$ of 
Eq.~\eqref{eq:fixed_input_dist}. Then, $4\nu_j(Z_j)F_j(C_jZ_j)$ is a PEF with power $\beta$ for the trial model 
$\cM_{j=2^k}$. 
\end{proposition}

\begin{proof}
In view of the model construction detailed in the paragraph including Eq.~\eqref{eq:fixed_input_dist},
the input-conditional output distributions according to the trial model $\cM_j$ constructed 
conditionally on $L\geq j$ form the same set $\cT$, independent of $j$. Since the input distribution 
$\nu_j(Z_j)$ according to $\cM_j$ is fixed, each distribution of $\cM_j$ can be expressed as 
$\mu_{j}(C_jZ_j)=\nu_j(Z_j)\mu(C_j|Z_j)$ where $\mu(C_j|Z_j)\in \cT$. In view of the definition in 
Eq.~\eqref{eq:pef_def}, a PEF with power $\beta$ for $\cM_j$ is a non-negative function 
$F_j: cz\mapsto F_j(cz)$  satisfying the linear inequality
\begin{equation*}
  \sum_{cz} \nu_j(Z_j=z) F_j(cz) \mu(C_j=c|Z_j=z)^{1+\beta}\leq 1, 
\end{equation*}
for each distribution $\mu(C_j|Z_j)\in \cT$. Considering that when $j=2^k$ the input probability is 
$\nu_j(Z_j=z)=1/4$ for each $z$, the statement in the proposition follows. 
\end{proof}
We remark that Prop.~\ref{obs1} applies to an arbitrary trial model $\cM$ as long as according to 
$\cM$ the input distribution $\nu(Z)$ is fixed with $\nu(Z=z)>0$ for all $z$ and the input-conditional 
output distributions form the set $\cT$. 

We can take advantage of Prop.~\ref{obs1} for constructing PEFs for all trial positions in a block 
by interpolation. For this purpose, let $\tilde\nu_{q}(Z)$ be a distribution of $Z$ parameterized 
by a positive number $q\in (0,4/3)$ such that $\tilde\nu_{q}(Z=00)=1-3q/4$ and 
$\tilde\nu_{q}(Z=z) = q/4$ if $z\ne 00$. Then $\nu_{j}(Z) =\tilde\nu_{q_j}(Z)$, where $q_{j}=1/(2^{k}-j+1)$.  
Define $\tilde F(q;CZ)=4\tilde\nu_{q}(Z)F_{2^{k}+1-1/q}(CZ)$ so that $\tilde F(q_{j};CZ)=4\nu_{j}(Z)F_{j}(CZ)$. 
Here, we implicitly allow non-integer suffixes for $F$, but we will not explicitly refer to non-integer
suffixed $F$. Below, $q$ is the parameter to be interpolated on. According to Prop.~\ref{obs1},
for each $j$, $\tilde F(q_{j};CZ)$ is constrained to be a PEF with power $\beta$ for the model $\cM_{j=2^k}$.
This fact motivates us to construct a function $\tilde F(q;CZ)$ with $q\in (0,4/3)$ such that this function 
is always a PEF with power $\beta$ for $\cM_{j=2^k}$. 
Suppose that both $\tilde F(q';CZ)$ and $\tilde F(q'';CZ)$ with $q'<q''$ are PEFs with power $\beta$
for $\cM_{j=2^k}$. According to the PEF definition, any convex combination of $\tilde F(q';CZ)$ 
and $\tilde F(q'';CZ)$ is also a PEF with power $\beta$ for $\cM_{j=2^k}$. Therefore, to ensure that 
the function $\tilde F(q; CZ)$ with $q\in [q',q'']$ is a PEF with power $\beta$ for $\cM_{j=2^k}$, 
we can construct this function as the linear interpolant 
\begin{equation}
\tilde F(q; CZ)= L_{q',q''}(q; CZ) \equiv \frac{q''-q}{q''-q'} \tilde F(q'; CZ)+ \frac{q-q'}{q''-q'} \tilde F(q''; CZ),
\end{equation} 
a convex combination of $\tilde F(q';CZ)$ and $\tilde F(q'';CZ)$. Considering that the parameter $q_j$
depends on $j$ monotonically, the above linear interpolation provides a way of constructing PEFs
$F_{j}(CZ)$ with power $\beta$ for the trial model $\cM_j$ as follows: Given two PEFs $F_{j_1}(CZ)$ 
and $F_{j_2}(CZ)$ with power $\beta$ for $\cM_{j_1}$ and $\cM_{j_2}$ respectively, we first 
compute the corresponding  $\tilde F(q_{j_1};CZ)$ and $\tilde F(q_{j_2};CZ)$, and then we can determine 
a PEF $F_{j}(CZ)$ with power $\beta$ for $\cM_{j}$, for any $j$ between $j_1$ and $j_2$, from the 
$\tilde F(q;CZ)$ constructed as the linear interpolant $L_{q_{j_1},q_{j_2}}(q; CZ)$ according to $F_{j}(CZ)=\tilde F(q_{j};CZ)/(4\nu_{j}(Z))$. 
To construct the PEFs with power $\beta$ for all trial 
positions in a block,  we first find the optimal PEFs witnessing $g_{b,j}(\beta)$ for three trial positions 
$j=1, j_{\text{mid}}, 2^k$ and compute the corresponding  $\tilde F(q_{j=1};CZ)$, $\tilde F(q_{j_{\text{mid}}};CZ)$ and 
$\tilde F(q_{j=2^k};CZ)$, where the choice of $j_{\text{mid}}$ can be optimized (see the next paragraph). Then
we construct the function $\tilde F(q; CZ)$ by linear interpolation as 
\begin{equation}
\tilde F(q; CZ)=   
    \left\{\begin{array}{ll}
          L_{q_{j=1},q_{j_{\text{mid}}}}(q; CZ) & \textrm{ if } q\in [q_{j=1}, q_{j_{\text{mid}}}],\\
          L_{q_{j_{\text{mid}}},q_{j=2^k}}(q;CZ) &\textrm{ if } q\in [q_{j_{\text{mid}}}, q_{j=2^k}].
        \end{array}\right.
\end{equation}
We refer to the above method as the linear-interpolation method, where the PEFs for the trial positions 
$j=1, j_{\text{mid}}, 2^k$ are optimized.

Using the trial-wise PEFs obtained by the above linear-interpolation method and in view of 
Eq.~\eqref{eq:block_gain_rate3}, we can compute a lower bound on the quantity $g'_b(\beta)$ with honest devices. 
We observe that there exists an optimal middle trial position $j_{\text{mid},\text{opt}}$ such 
that the computed lower bound on $g'_b(\beta)$ is as high as possible. The optimal position 
$j_{\text{mid},\text{opt}}$ depends not only on the honest-device distribution $\nu_{h}(C|Z)$ 
 but also on the power $\beta$ of trial-wise PEFs, and it can be found by a line search. 
Furthermore, we observe that the trial-wise PEFs obtained by linear interpolation with 
the optimal position $j_{\text{mid},\text{opt}}$ witness a lower bound of $g'_b(\beta)$ which is at least 
$99.99\%$ of $g'_b(\beta)$. Therefore, the trial-wise PEFs obtained by linear interpolation
can perform almost as well as the optimal trial-wise PEFs that witness the value of $g'_b(\beta)$. 
In our numerical analysis for randomness expansion detailed in the next three sections, we used the above 
linear-interpolation method with $j_{\text{mid},\text{opt}}=53,478$ where the maximum block length is $2^{17}$.
In this way, we need only to find the optimal PEFs for three trial positions in order to construct a well-performing
PEF for a block with length $2^{17}$.

\section{Protocol design and commissioning}
\label{sect:commission}

Protocol design and commissioning consists of choosing the protocol parameters based on anticipated 
experiment performance. For our demonstration, the first task was to pick the maximum block length
that was used in our experiment. This had to be done before acquiring the data 
because it affects the experiment itself. The remaining parameters were determined after the data 
was acquired but before the data was unblinded. However, for production-quality implementations, 
all protocol parameters should be chosen before acquiring data. 

\subsection{Block-length determination}
\label{subsect:block_length}

Our randomness-expansion experiment was performed in August of 2018. Before that, in July of 2018, we acquired about $8$ minutes of experimental Bell-trial data at the rate of approximately $100,000$ trials per second, where the raw counts are summarized in Table~\ref{tab:exp_design_raw}. This data was used to determine the optimal choice for the maximum block length of our randomness-expansion experiment under the assumption that the devices perform as inferred from this data.  We estimated the input-conditional output distribution for a trial by maximum likelihood (as detailed in Sect.~VIII B of Ref.~\cite{knill:qc2017a}). The estimate is denoted by $\nu(C|Z)$ and shown in Table~\ref{tab:exp_design_dist}.  
For each possible choice $2^k$ for the maximum block length, we formulated the following optimization problem:  Assuming that the quantum devices used are honest with the input-conditional output distribution $\nu(C|Z)$ for 
each trial, minimize the number of blocks required for randomness expansion with soundness error $\epsilon$. Denote
this minimum number of blocks by $N_{b,\min}(k)$, where the dependence on $k$ is explicit and the dependence on 
$\epsilon$ is implicit. The minimum number of blocks $N_{b,\min}(k)$ can be found by a binary search 
provided that for each $N_{b}$ and $k$ we can determine whether $N_{b}$ blocks suffice for randomness 
expansion with maximum block length $2^{k}$.  See the steps of Algorithm~\ref{alg:expans_opt} for details. 
Once $N_{b,\min}(k)$ is known,  we can choose $k$ to minimize the expected number of trials 
$N_{t,\min}(k)=N_{b,\min}(k)\times (1+2^k)/2$ in $N_{b,\min}(k)$ blocks. 
Define $k_{\text{opt}}$ to be the minimizing value of $k$. The minimum number of trials $N_{t,\min}(k)$ required for randomness expansion when $k=k_{\text{opt}}$ is abbreviated as $N_{t,\min}$.

\begin{table}[htb!]
 \caption{Counts of measurement settings $xy$ and outcomes $ab$ used for finding the optimal choice for the maximum block length.}\label{tab:exp_design_raw}
 \begin{tabular}{|ll|l|l|l|l|}
    \hline
    &$ab$&00&10&01&11\\
    $xy$&&&&&\\
    \hline
    00&&
    11183694 &   11345 &      12229 &      28730  \\
    10&&
    11092860  &      98100 &      10996  &     29439 \\
    01&&
    11094694  &     11817  &     98213   &   27771 \\
    11&&
    10982482  &    125705  &    123749   &     2306 \\
    \hline
  \end{tabular}
\end{table}

\begin{table}[htb!]
\caption{The input-conditional output distribution $\nu(C|Z)$ with $C=AB$ and $Z=XY$ by maximum likelihood using the raw data in 
Table~\ref{tab:exp_design_raw}. They are used for determining the optimal choice for the maximum block length, not to make a statement about the actual distribution when running the experiment. }
\label{tab:exp_design_dist}
 
 \begin{tabular}{|ll|l|l|l|l|}
    \hline
    &$ab$&00&10&01&11\\
    $xy$&&&&&\\
    \hline
    00&&
    0.995376279105447  &  0.001005120002272 &  0.001083852982446 &  0.002534747909835 \\
    10&&
    0.987636235767830  & 0.008745163339888  & 0.000983148485219  & 0.002635452407063 \\
    01&&
    0.987719250485132  & 0.001056875186745  & 0.008740881602761  & 0.002482992725362 \\
    11&&
    0.977600576729569  & 0.011175548942308  & 0.011018807523479  & 0.000205066804643 \\
    \hline
  \end{tabular}
\end{table}

\SetAlgorithmName{Algorithm}{}{}
\RestyleAlgo{boxruled}
\vspace*{\baselineskip}
\begin{algorithm}
  \caption{Determine whether $N_{b}$ blocks suffice for randomness expansion with maximum block length $2^{k}$ and soundness error $\epsilon$, given that the input-conditional output distribution for each trial is $\nu(C|Z)$.}\label{alg:expans_opt}
\begin{minipage}{\textwidth-1in}
  \begin{enumerate}
  \item Maximize the expected net number of random bits $\sigma_{\text{net},\text{opt}}(\beta)$ over $\beta>0$ by 
  a local search, where $\sigma_{\text{net},\text{opt}}(\beta)$ is computed as follows:
    \begin{enumerate}
    \item Determine the trial-wise PEFs with power $\beta$ according to the linear-interpolation method of 
    Sect.~\ref{sect:pef_constrct_impl}, where the choice for the middle trial position $j_{\text{mid}}$ 
    required is optimized.
    \item Compute the randomness rate per block $g_{b}(\beta,2^{k})$ at
      $\nu(C|Z)$ witnessed by these PEFs.
    \item Maximize $\sigma_{\text{net}}(\beta, \epse)$ over $\epse$ by a line search, where $\sigma_{\text{net}}(\beta, \epse)$  is the expected net number  of random bits at entropy error $\epse$ computed as follows:
      \begin{enumerate}
      \item Compute the adjusted output entropy witness $\sigma_{\text{in}}$ expected after $N_b$ blocks, $\sigma_{\text{in}}=N_b g_b(\beta, 2^k)+\log_2(\epse)/\beta+\log_2(\epsilon)$ (Eq.~\eqref{eq:critical_cond}). 
      \item Set the extractor error as $\epsx=\epsilon-\epse$.
      \item Determine the number $k_{\text{out}}$ of extractable random bits and the number $d_{\text{s}}$ of seed bits required according to the extractor constraints of Eq.~\eqref{eq:tmps_con} with $m_{\text{in}} = N_b\times 2^k\times 2$. Note: The  $\cX$ of Eq.~\eqref{eq:chi_set} is guaranteed to be non-empty. 
      \item Set $\sigma_{\text{net}}(\beta, \epse)=k_{\text{out}}-d_{\text{s}}-N_b(k+2)$. Note:
        Each block consumes $k$ random bits for determining its length        and $2$ random bits for the settings choices of the spot-checking trial.
      \end{enumerate}
     \item Set the maximum of $\sigma_{\text{net}}(\beta, \epse)$ over $\epse$ as $\sigma_{\text{net},\text{opt}}(\beta)$.
    \end{enumerate}
  \item Set $\beta_{\text{opt}}$ to be the $\beta$ that maximizes $\sigma_{\text{net},\text{opt}}(\beta)$.
  \item If $\sigma_{\text{net},\text{opt}}(\beta_{\text{opt}})\geq 0$,
    then randomness expansion with parameters $(N_b, 2^k, \epsilon)$ at $\nu(C|Z)$ is possible.
  \end{enumerate}
\end{minipage}

\end{algorithm}

With the procedure outlined in the previous paragraph and the 
input-conditional output distribution $\nu(C|Z)$ given in Table~\ref{tab:exp_design_dist}, we determined the minimum  number of trials $N_{t,\min}$ required and the associated optimal choice $2^{k_{\text{opt}}}$ for the maximum block length, in order to achieve randomness expansion with a soundness error $\epsilon$ varying from $10^{-3}$ to $10^{-12}$. Here we consider the case where the inputs $X=0$ and $Y=0$ are used in each non-spot-checking trial. 
The results are summarized in Table~\ref{tab:opt_block_length}. Several interesting points illustrated by the results in Table~\ref{tab:opt_block_length} are as follows: 1) The optimal choice for the maximum block length is $2^{17}$, independent of the soundness error. 2) Both the optimal power $\beta_{\text{opt}}$ and the optimal error-splitting ratio (that is, the ratio of the optimal entropy error to the optimal extractor error) are independent of the soundness error. 3) The minimum number of trials $N_{t,\min}$ required for randomness expansion scales linearly with the logarithm of the soundness error. In addition,  we observed that the optimal choice for the maximum block is independent of the particular  inputs used in every non-spot-checking trial. However, the optimal choice for the maximum block length, as well as the optimal PEF power and the optimal error-splitting ratio, depends on the input-conditional output distribution $\nu(C|Z)$.   Particularly, we observed that these optimal parameters are well correlated with the statistical strength for rejecting local realism, which is the minimum Kullback-Leibler divergence of the true distribution $\nu(C|Z)/4$ of 
Bell-trial results from the local realistic distributions~\cite{vandam:qc2003a, zhang_y:qc2010a} supposing that the inputs of Bell trials are uniformly randomly distributed. For the distribution $\nu(C|Z)$ given in 
Table~\ref{tab:exp_design_dist}, the statistical strength for rejecting local realism is $7.19 \times 10^{-6}$.

\begin{table}
  \caption{Parameters for achieving randomness expansion with soundness error $\epsilon$ at the distribution 
  $\nu(C|Z)$  given in Table~\ref{tab:exp_design_dist}. The fixed inputs $X=0$ and $Y=0$ are used in each 
  non-spot-checking trial. Here $N_{t,\min}$ is the minimum number of trials required, $2^{k_{\text{opt}}}$ 
  is the optimal choice for the maximum block length, $\beta_{\text{opt}}$ is the optimal power of the trial-wise 
  PEFs used, and $\epsilon_{\text{en}, \text{opt}}$  is the associated optimal entropy error.}
 \label{tab:opt_block_length}
  \begin{equation}
    \begin{array}{|l|l|l|l|l|l|}
      \hline
       \epsilon & N_{t,\min} &2^{k_{\text{opt}}} &\beta_{\text{opt}}& \epsilon_{\text{en}, \text{opt}}\\
      \hline
      1\times 10^{-3}&  1.90\times10^{11} & 2^{17}& 1.32\times10^{-7} &  9.78\times 10^{-4} \\

      1\times 10^{-6}&  3.80\times10^{11} & 2^{17}& 1.32\times10^{-7} &  9.78\times 10^{-7} \\ 

      1\times 10^{-9}&  5.70\times10^{11} & 2^{17}& 1.32\times10^{-7} &  9.78\times 10^{-10} \\

      1\times 10^{-12}& 7.60\times10^{11} & 2^{17}& 1.32\times10^{-7} &  9.78\times 10^{-13} \\
      \hline
    \end{array}
    \notag
  \end{equation}  
\end{table}

\subsection{Parameter determination}
\label{subsect:para}

After choosing $2^{17}$ as the maximum block length, we ran the randomness-expansion experiment and collected $110.3$ hours worth of data over the course of two weeks. The data was acquired in a series of cycles.  Each cycle began with about 2 minutes of calibration trials, generated at the rate of approximately $250,000$ trials per second, which were stored in a calibration file, and then proceeded with a varying number of expansion files. Each expansion file recorded about 2 minutes of block data, generated at the rate of approximately $153$ blocks per second, which were followed by about $5$ seconds of calibration trials, generated at the rate of approximately $250,000$ trials per second and recorded at the end of the file.  For calibration trials, the input settings were chosen uniformly. Note that pseudorandom settings choices would suffice for calibration purposes.  For non-spot-checking trials in a block the inputs were fixed to be $X=0$ and $Y=0$. 
Non-spot-checking trials with no detections, namely those satisfy $a =0$ and $b = 0$, were not explicitly recorded. For non-spot-checking trials with detections and for the spot-checking trial in a block, their positions  and  outcomes, as well as the settings choices used in the spot-checking trial, were recorded. Because the probability of detections at inputs $X=0$, $Y=0$ is less than $0.0046$, this recording method saves space. 

For commissioning and training purposes, we  unblinded the
first $16$ cycles, which contains $4,502,276$ blocks (about
$7.4\%$ of the recorded block data). We refer to this data as
  the training set and the remaining blinded data as the analysis
  set. We use the training set to choose the protocol input
  parameters $(k_{\text{out}}$, $\epsilon)$ and the required
  pre-analysis-run parameters $(\beta$, $t_{\min}$, $\epse)$, as well as
  the associated extractor parameters.  To make our choices, we took
  advantage of prior knowledge of the number of blocks $N_b$
  available in the analysis set. We found that $N_b=56,070,910$.
  In addition we were aware of specific
  properties of the analysis set such as which cycles had  
  reduced-length calibration files. We remark that in 
  production-quality implementations, such
  specific knowledge is not available.

  First, we fixed the soundness error to be $\epsilon=5.7 \times 10^{-7}$,
  corresponding to the 5-sigma criterion. We estimated the
  input-conditional output distribution $\nu(C|Z)$ by maximum
  likelihood (Sect.~VIII B of Ref.~\cite{knill:qc2017a}) from the
  calibration trials recorded in the first $6$ cycles of the training
  set.  We used only the first $6$ cycles for this purpose because
  there were indications that they were representative of a stable
  setup. The counts used to infer $\nu(C|Z)$ are shown in
  Table~\ref{tab:commission_raw} and the values of $\nu(C|Z)$ inferred
  are in Table~\ref{tab:commission_dist}.  We then determined the PEF power
  $\beta$ and entropy error $\epse$ by setting them to the optimizing
  quantities in Algorithm~\ref{alg:expans_opt} with input parameters 
  $N_{b}=56,070,910$ and $\epsilon=5.7 \times 10^{-7}$, given that 
  the maximum block length is $2^{17}$ and the input-conditional 
  output distribution for each trial is $\nu(C|Z)$.  We obtained  
  $\beta=4.7614\times 10^{-8}$  and  $\epse=5.6822\times 10^{-7}$. When running
  Algorithm~\ref{alg:expans_opt} we also obtained the trial-wise PEFs
  corresponding to the above $\beta$, according to which the
  randomness rate per block is $36.06$ bits.  
  We remark that the distribution in Table~\ref{tab:commission_dist}
  is different from the distribution in Table~\ref{tab:exp_design_dist}
  determined from data acquired before the experimental configuration was finalized.
  Particularly, the statistical strength for rejecting local realism of
  the distribution in Table~\ref{tab:commission_dist} is $3.03 \times 10^{-6}$, 
  much lower than the statistical strength $7.19 \times 10^{-6}$ of 
  the distribution in Table~\ref{tab:exp_design_dist}. Therefore, 
  the parameters $\beta$ and $\epse$ found above are different from
  what were obtained in the original run of Algorithm~\ref{alg:expans_opt} 
  which resulted in our choosing  $2^{17}$ as the maximum block length.
  In the original run of Algorithm~\ref{alg:expans_opt}, if we had used 
  the distribution in Table~\ref{tab:commission_dist} determined from the 
  calibration trials recorded in the first $6$ cycles,  we would  have 
  chosen $2^{18}$ as the maximum block length.

Second, we need to choose the success threshold $t_{\min}$ for running Protocol~\ref{prot:condgen_direct}, or
equivalently,  the success threshold $W_{\min}$ for the output entropy witness stated in the main text. In view of the success condition $\prod_{i=1}^{N_b}G_{i}(\Sfnt{C}_{i}\Sfnt{Z}_{i}) \geq t_{\min}$ in Protocol~\ref{prot:condgen_direct} and the definition of the output entropy witness below Eq.~\eqref{eq:threhd_cond}, $W_{\min}$ is related with $t_{\min}$  
by $W_{\min}=\log_2(t_{\min})/\beta$. To determine the value for $W_{\min}$ used in our expansion analysis, we studied 
the dependence of the success probability in an honest implementation of the protocol on the threshold $W_{\min}$. Since the number of blocks available for expansion analysis is $N_b=56,070,910$ and the randomness rate per block estimated  
in the previous paragraph is $g_b=36.06$ bits, the output entropy witness after $N_b$ blocks with honest devices, where
the input-conditional output distribution is fixed to be the $\nu(C|Z)$ in Table~\ref{tab:commission_dist}, is expected to be $N_b g_b=2,021,917,014$ bits. To estimate the success probability, we also need to know the variance 
$\sigma_{v}^2$ of the output entropy witness after $N_b$ blocks. The variance is given as $\sigma_{v}^2=N_b \Var_{\nu} \big(\log_2(G(\Sfnt{C}\Sfnt{Z}))/\beta\big)$, where $\Var_{\nu} \big(\log_2(G(\Sfnt{C}\Sfnt{Z}))/\beta\big)$ is the variance of the variable $\log_2(G(\Sfnt{C}\Sfnt{Z}))/\beta$ according to the distribution $\nu(\Sfnt{C}\Sfnt{Z})$ determined by the honest-device input-conditional output distribution $\nu(C|Z)$. 
By the results presented in the third paragraph of Sect.~\ref{sect:pef_constrct_impl}, we found that $\Var_{\nu} \big(\log_2(G(\Sfnt{C}\Sfnt{Z}))/\beta\big)=4.6729 \times 10^{8}$ and so $\sigma_{v}^2=2.6201\times 10^{16}$. 
In view of the central limit theorem, we assume that the output entropy witness after $N_b$ blocks with honest devices is normally distributed with mean $N_b g_b$ and variance $\sigma_{v}^2$.
Thus, given that the success threshold for the output entropy witness is $W_{\min}$, the success probability is  estimated to be $p_{\mathrm{succ}}=Q\big(-(N_b g_b-W_{\min})/\sigma_{v}\big)$. Here the function $Q$  is the tail distribution function of the standard 
normal distribution. For our expansion analysis, we chose the success probability $p_{\mathrm{succ}}=0.9938$ such 
that $(N_b g_b-W_{\min})/\sigma_{v}=2.5$, matching the conventional one-sided 2.5-sigma criterion.  
Consequently, we determined that $W_{\min}=1,616,998,677$ and so $t_{\min}=2^{77}$. 
The completeness calculation just performed is heuristic in that
we assume that the output entropy witness at the end is sufficiently
close to normally distributed for the tail calculation to be accurate.
On the other hand, it is somewhat pessimistic because it does not account for
the possibility that the threshold is reached early but is not exceeded at 
the end in the absence of an early stop.

 Third, we determined the number of seed bits required for applying the TMPS extractor. 
 Considering that $t_{\min}= 1/(q^\beta\epse)$ with  $q=2^{-\sigma_{\text{in}}}\epsilon$ in 
 Protocol~\ref{prot:condgen_direct}, we set $\sigma_{\text{in}}=\log_2(t_{\min})/\beta+\log_2(\epse)/\beta+\log_2(\epsilon)=1,181,264,480$. 
Moreover, we set the length in bits of the extractor input to be $m_{\text{in}}=N_b\times 2^{17}\times 2=14,698,652,631,040$ and the extractor error to be $\epsx=\epsilon-\epse=1.7800\times 10^{-9}$. 
 The condition $\sigma_{\text{in}}\leq m_{\text{in}}+ \frac{1+\beta}{\beta} \log_2(\epsilon)$
required for defining the set $\cX$ of Eq.~\eqref{eq:chi_set} is satisfied.  Therefore, according to the TMPS extractor 
constraints in Eq.~\eqref{eq:tmps_con}, $d_{\text{s}}= 3,725,074$ seed bits are needed, and conditional on success $k_{\text{out}}=1,181,264,237$ new random bits can be extracted.  As the protocol is designed to
 consume $k_{\text{in}}=1,069,072,364$ random bits, including $1,065,347,290$ random bits for 
 spot checks and settings choices as well as $3,725,074$ seed bits, the expected expansion ratio conditional on success according
to this calculation is $k_{\text{out}}/k_{\text{in}}=1.105$.  However,
if the threshold for success is reached early, the expansion ratio
is higher, as witnessed by the final results of our protocol run.

 In the same way as above, we can vary the value of the desired success probability $p_{\mathrm{succ}}$ and 
 study the dependence of the expected expansion ratio at the soundness error 
 $\epsilon=5.7 \times 10^{-7}$ on $p_{\mathrm{succ}}$  (see Fig. 3 of the main text).  Moreover, we can also 
 vary the value of the desired soundness error $\epsilon$ and study the dependence of the expected expansion 
 ratio on $\epsilon$ given the fixed success probability $p_{\mathrm{succ}}=0.9938$. 
 The results are illustrated in Fig.~\ref{fig:varying_soundness}. 
 
 \begin{figure}[htb!]
  \begin{center}
   \includegraphics[scale=0.55,viewport=6cm 9cm 14.5cm 21cm]{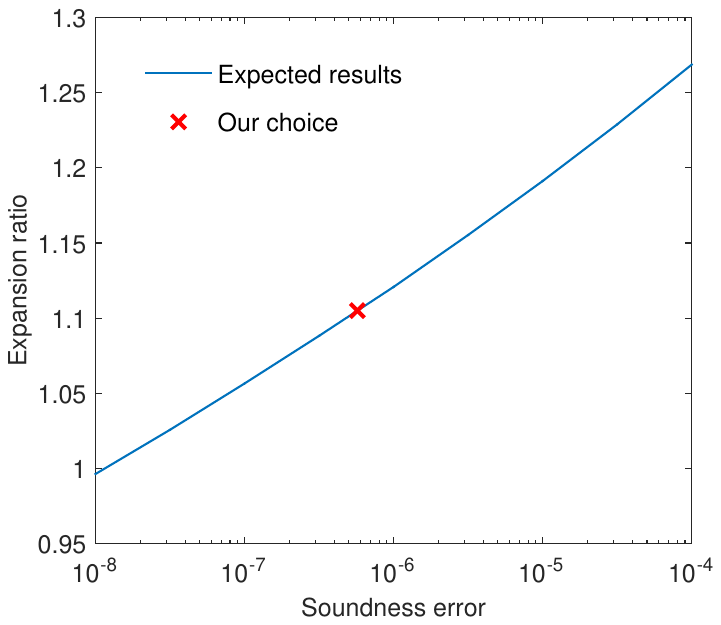}
  \end{center}
  \caption{The expected expansion ratio as a function of the soundness error, when fixing the desired 
  success probability to be $0.9938$.   Given the input-conditional output distribution $\nu(C|Z)$ of 
  Table~\ref{tab:commission_dist} expected for each trial, we first optimized over the PEF power $\beta$ 
  and the entropy error $\epse$ in order to maximize the expected net number of random 
  bits extractable from $N_{b}=56,070,910$ blocks of analysis data at a chosen soundness error 
  (according to Algorithm~\ref{alg:expans_opt}). Then, with the optimal solutions found above we computed  
  the expected expansion ratio (see the text for details), supposing that the success probability is 
  $0.9938$. Our choice for the soundness error and the corresponding expected expansion ratio is labelled 
  by the cross.}
  \label{fig:varying_soundness}
\end{figure}

\begin{table}[htb!]
 \caption{Counts of measurement settings $xy$ and outcomes $ab$ in the calibration trials collected over the first 6 unblinded cycles.}\label{tab:commission_raw}
 \begin{tabular}{|ll|l|l|l|l|}
    \hline
    &$ab$&00&10&01&11\\
    $xy$&&&&&\\
    \hline
    00&&
    62824397  &     64859   &    71896    &  153039 \\
    10&&
    62360267  &    524745   &    60696    &  165193 \\
    01&&
    62385836  &     64579   &   506557    &  153310 \\
    11&&
    61772852  &    672142   &   642597    &   16105 \\
    \hline
  \end{tabular}
\end{table}

\begin{table}[htb!]
\caption{The input-conditional output distribution $\nu(C|Z)$ with $C=AB$ and $Z=XY$ by maximum likelihood using the raw data in 
Table~\ref{tab:commission_raw}. They are used for determining the PEF power $\beta$ and the 
entropy error $\epse$ for expansion analysis, not to make a statement about the actual distribution 
when running the experiment.} 
\label{tab:commission_dist}
 
 \begin{tabular}{|ll|l|l|l|l|}
    \hline
    &$ab$&00&10&01&11\\
    $xy$&&&&&\\
    \hline
    00&&
    0.995404388386381  & 0.001026519904493  & 0.001141638426253  & 0.002427453282873 \\
    10&&
    0.988123719866393  & 0.008307188424481  & 0.000959649740469  & 0.002609441968657 \\
    01&&
    0.988527138911423  & 0.001024397388010  & 0.008018887901211  & 0.002429575799356 \\
    11&&
    0.978890595338359  & 0.010660940961074  & 0.010192774268503  & 0.000255689432064 \\
    \hline
  \end{tabular}
\end{table}

\section{PEF updating during the analysis}
\label{sect:calib}

Based on the training set, the need for periodic realignment
  during the experiment and the reports from the experimenters, we
  anticipated that the input-conditional output distribution drifted
  significantly during the experiment. We therefore decided to
  update block-wise PEFs used for each next cycle based on calibration data
  preceding the block data of the cycle. We decided to always use
  at least $n_{\text{calib},\min}=22,200,000$ calibration trials for this purpose. 
  For $14$ of the $140$ cycles of the analysis set, the cycle's calibration file
  did not contain sufficiently many trials. For these cycles, we used
  also calibration trials from the last expansion files of the previous
  cycle to obtain at least $n_{\text{calib},\min}$ calibration trials in total. From the 
  calibration data, we determined the (presumed current) honest-device 
  input-conditional output distribution  $\nu(C|Z)$ by maximum likelihood
  (Sect~ VIII B of Ref.~\cite{knill:qc2017a}) and obtained the
  trial-wise PEFs with power $\beta=4.7614\times 10^{-8}$ 
  for all possible trial positions in a block by the linear-interpolation method
  (Sect.~\ref{sect:pef_constrct_impl}). The block-product of these PEFs
  is a PEF with the same power $\beta$  for the block.  
  To verify that PEFs obtained in this way achieve close to optimal performance
  for the true distribution, we performed simulations.  The simulation
  and its results are described in Fig.~\ref{fig:resamed_pefs} and
  show that the PEFs obtained perform close to optimal with high
  probability.

\begin{figure}[htb!]
  \begin{center}
   \includegraphics[scale=0.55,viewport=6cm 8.5cm 14.5cm 21.5cm]{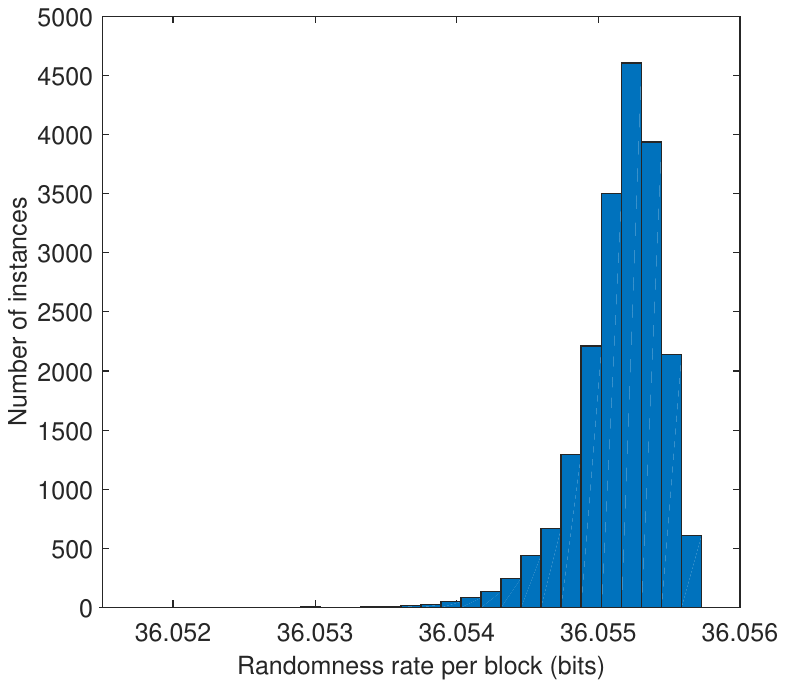}
  \end{center}
  \caption{Performance of PEFs constructed using random samples. We assumed that the true distribution 
  of calibration trials is $\nu(C|Z)/4$, where $\nu(C|Z)$ is the input-conditional output distribution 
  determined from the calibration trials collected over the first 6 unblinded cycles 
  (Table~\ref{tab:commission_dist}). From this true distribution, we determined that the 
  randomness rate per block is $36.0558$ bits (given the PEF power $\beta=4.7614\times 10^{-8}$). 
  Then we drew $20, 000$ random samples according to the true distribution, where each sample 
  has $30,000,000$ trials corresponding to $2$ minutes of calibration trials at the rate of 
  $250,000$ trials per second.  For each sample, we determined a block-wise PEF with power 
  $\beta=4.7614\times 10^{-8}$ following the same procedure as described in the first 
  paragraph of Sect.~\ref{sect:calib}. Then we computed the randomness rate per block 
  witnessed by this block-wise PEF.  For this computation, we assumed that each trial in 
  a block has the same input-conditional output distribution $\nu(C|Z)$ according to the true
  distribution of sampled results. The results demonstrate that the randomness rate per 
  block witnessed is almost independent of the random sample used. This observation 
  justifies that it is sufficient to update the block-wise PEF using only $2$ minutes 
  of calibration trials.}
  \label{fig:resamed_pefs}
\end{figure}

  Before running the protocol on the analysis set, we tested the
  performance of the protocol with cycle-updated PEFs on the training
  set. The training set contains $4,502,276$ blocks. The empirical 
  randomness rate per block from this test run is $35.6851$ bits,
  consistent with the prediction of $36.0558$ bits based on 
  the input-conditional output distribution determined from the first
  $6$ cycles in the training set (Table~\ref{tab:commission_dist}). 
  The empirical randomness rate per block was obtained by dividing 
  the final value of the running output entropy witness by the number 
  of blocks processed. The empirical variance of the per-block output 
  entropy witnesses is $4.4786\times 10^{8}$,  similar to the predicted 
  variance of $4.6729\times 10^{8}$.

\section{Analysis results}
\label{sect:anlys}

Our implementation of Protocol~\ref{prot:condgen_direct} for randomness 
analysis is shown in Protocol~\ref{prot:condgen_direct_implt}. 
There were two independent runs of expansion analysis. The primary one was performed using 
MATLAB which took $45.4$ hours on a personal computer, while the secondary one used Python 
which took also about $45$ hours on a personal computer. The main purpose of the secondary run 
was as a consistency check during training and analysis.  The two independent runs returned 
consistent analysis results. The results reported here correspond to the primary run.

\SetAlgorithmName{Protocol}{}{}
\RestyleAlgo{boxruled}
\vspace*{\baselineskip}
\begin{algorithm}[H]
  \caption{Protocol as implemented.}\label{prot:condgen_direct_implt}
  \begin{minipage}{\textwidth-1in}
  
  \Input{
  \begin{itemize}
  \item Analysis Data---calibration and expansion data acquired successively 
   in a series of cycles. Each cycle begins with a calibration file 
   and then proceeds with a varying number of expansion files. Each expansion 
   file records a varying number of blocks, which are followed by a few seconds 
   of calibration trials.
  \end{itemize} 
  }

  \Given{
  \begin{itemize}
  \item $W_{\min}$---the success threshold for the running output entropy witness (chosen to be $1,616,998,677$, 
  see Sect.~\ref{subsect:para}). 
  \item $N_{b}$---the number of blocks available in the analysis set (set to be $56,070,910$, see 
  Sect.~\ref{subsect:para}).
  \item $\beta$---the PEF power (chosen to be $4.7614\times 10^{-8}$, see 
  Sect.~\ref{subsect:para}). 
  \item $j_{\text{mid}}$---the middle trial position used by the linear-interpolation method (fixed to be $53,478$, 
  see Sect.~\ref{sect:pef_constrct_impl}). 
  \item $n_{\text{calib},\min}$---the minimum number of calibration trials to be used for PEF updating (set to 
  be $22,200,000$, see Sect.~\ref{sect:calib}).
  \end{itemize}  
  }

  \Output{$P$--- a binary flag indicating success ($P=1$) or failure ($P=0$).}
 Initialize the binary flag as $P=0$; \\
 Initialize the running output entropy witness as $W_{\text{run}} = 0$;\\
 Initialize the running number of blocks processed $N_{\text{run}} = 0$;
   \BlankLine
  
    \For{\text{cycle in Analysis Data}}{
    $n_{\text{calib},\text{act}}\gets$ \text{actual number of trials in the calibration file}\;
    $\{\text{Data}_{\text{calib}}(cz)\}_{cz}\gets$ \text{counts of measurement settings $z=xy$ and outcomes $c=ab$ in the calibration file};
    \If{$n_{\text{calib},\text{act}} < n_{\text{calib},\min}$}
     {Load the last expansion file in the previous cycle\;
     \While{$n_{\text{calib},\text{act}} < n_{\text{calib},\min}$}
    {Increment  $n_{\text{calib},\text{act}}$ by the number of calibration trials in the loaded expansion file\;
    Increment  $\{\text{Data}_{\text{calib}}(cz)\}_{cz}$ with the counts of settings and outcomes of calibration 
    trials in the loaded expansion file\;
    Load the previous expansion file\;}
     }    
     Use $\{\text{Data}_{\text{calib}}(cz)\}_{cz}$ to determine an honest-device input-conditional output distribution 
     $\nu(C|Z)$ by maximum likelihood (see Sect.~VIII B of Ref.~\cite{knill:qc2017a})\;
     Determine the trial-wise PEFs $\{F_j(C_jZ_j)\}_{j=1}^{2^{17}}$ with power $\beta$ for all possible trial positions 
     $j$ in a block according to the 
     linear-interpolation method of Sect.~\ref{sect:pef_constrct_impl}\;   
     \For{\text{expansion file in cycle}}
     {\For{\text{block in expansion file}}
     {$N_{\text{run}} \gets N_{\text{run}} +1$;\\   
     \For{\text{trial $j$ in block}}{
     Update $W_{\text{run}} \gets W_{\text{run}} + \log_{2}{(F_{j}(c_{j} z_{j}}))/\beta$, where $c_j$ and $z_j$ are the 
     settings and outcomes observed at the $j$'th trial;\\
     }
     }
      \If{$W_{\text{run}} \geq W_{\text{min}}$}
     {Record the number of blocks required for successful expansion as $n_b=N_{\text{run}}$\;
     Return $P=1$\tcp*{Protocol succeeded.} 
     }
     }    
     }
 \end{minipage}
\end{algorithm}

We processed the $56,070,910$ blocks of the $140$ cycles in
the analysis set successively. After processing $49,977,714$ blocks
corresponding to running the experiment for $91.0$ hours, 
the running output entropy witness exceeded the success threshold $W_{\min}$.
At this point, we consumed $949,576,566$ random bits for spot checks and 
settings choices. If we applied the TMPS extractor with the chosen
extractor parameters, we would consume an  additional $3,725,074$
random bits for seed and output $1,181,264,237$ bits with soundness error
$5.7 \times 10^{-7}$ for an actual expansion ratio of $1.24$. 
We continued processing the remaining blocks and accumulating
the running output entropy witness. After processing all blocks, we
observed that the empirical randomness rate per block witnessed by
cycle-updated block-wise PEFs with power $\beta=4.7614\times 10^{-8}$ is
$32.8028$ bits, lower than the randomness rate $36.0558$ bits per
block predicted by the input-conditional output distribution $\nu(C|Z)$ of
Table~\ref{tab:commission_dist}. The empirical variance of the per-block 
output entropy witnesses is $4.3264\times 10^8$, also lower than 
the predicted variance of $4.6729\times 10^{8}$. The complete
dynamics of the adjusted output entropy witness is illustrated in
Fig. 4 of the main text.  During the expansion analysis, we also
tracked the drifts of the randomness rate per block as well as the
statistical strength for rejecting local realism~\cite{vandam:qc2003a,
  zhang_y:qc2010a}, see Fig.~\ref{fig:addi_anlys}. These results
suggest that the randomness rate per block and statistical strength
are positively correlated.

\begin{figure}[htb!]
  \begin{center}
   \includegraphics[scale=0.6,viewport=6cm 8cm 14.5cm 21.5cm]{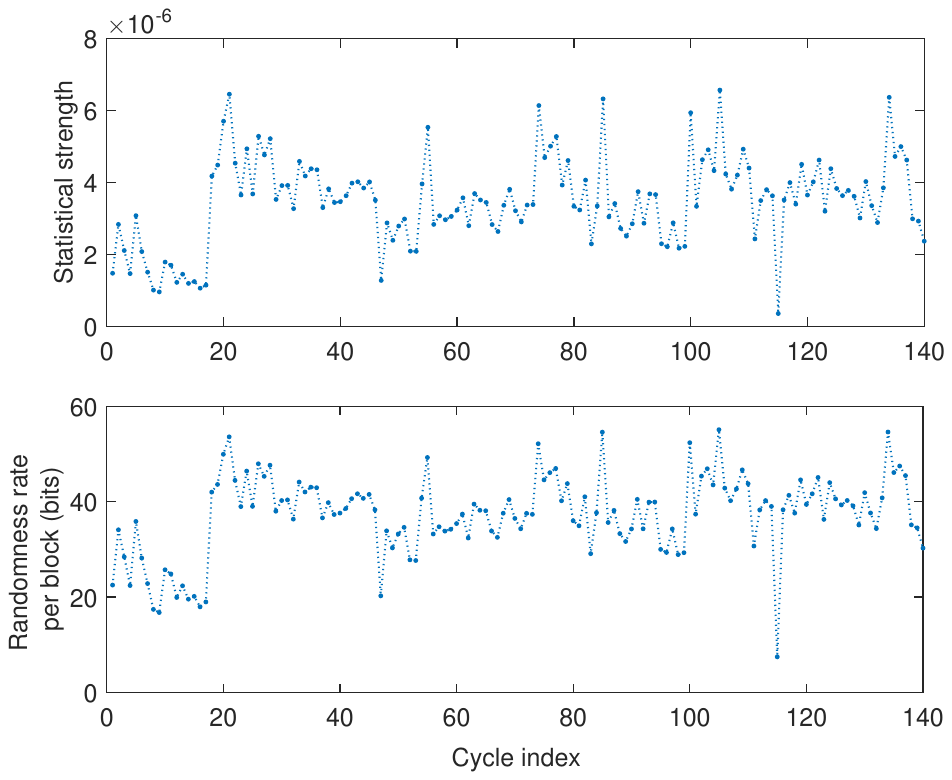}
  \end{center}
  \caption{The statistical strength and randomness rate per block predicted by the calibration data 
  used in each cycle. Based on 
  at least $22,000,000$ of the most recent calibration trials preceding the blocks in each cycle, 
  we estimated the input-conditional output distribution for a trial. Assuming that the estimated 
  distribution is the true one for all trials of the next cycle, we computed the statistical strength
  for rejecting local realism~\cite{vandam:qc2003a, zhang_y:qc2010a} and randomness rate per block 
  witnessed by the cycle-updated block-wise PEF with power $\beta=4.7614\times 10^{-8}$.}
  \label{fig:addi_anlys}
\end{figure}

\end{document}